\documentclass[aps,prx,twocolumn,superscriptaddress,floatfix,
  nofootinbib]{revtex4-2}
\usepackage{amsmath,amssymb,amsthm}
\usepackage{graphicx}
\usepackage{booktabs}
\usepackage{xcolor}
\usepackage{tikz}
\usetikzlibrary{arrows.meta,positioning,decorations.pathreplacing}
\usepackage{pgfplots}
\pgfplotsset{compat=1.18}

\newtheorem{theorem}{Theorem}
\newtheorem{lemma}{Lemma}
\newtheorem{corollary}{Corollary}
\newtheorem{definition}{Definition}
\newtheorem{remark}{Remark}

\newcommand{\Abfk}{\mathcal{A}}                       
\newcommand{\Rbfk}[2]{\mathcal{R}_{\mathrm{BFK}}(#1,#2)} 
\newcommand{\CCZ}{\mathrm{CCZ}}
\newcommand{\ket}[1]{|#1\rangle}
\newcommand{\bra}[1]{\langle #1|}

\begin{document}

\title{BPBO: Blindness-Preserving Brickwork Optimization
by Certified Region Resynthesis}

\author{Youngkyung Lee}
\email{youngklee@etri.re.kr}
\affiliation{Cryptography Engineering Laboratory, Electronics and Telecommunications Research Institute, Daejeon, South Korea}
\author{Juyoung Kim}
\email{ap424@etri.re.kr}
\affiliation{Cryptography Engineering Laboratory, Electronics and Telecommunications Research Institute, Daejeon, South Korea}
\author{Doyoung Chung}
\email{thisisdoyoung@etri.re.kr}
\affiliation{Cryptography Engineering Laboratory, Electronics and Telecommunications Research Institute, Daejeon, South Korea}
\date{June 29, 2026}

\begin{abstract}
Universal blind quantum computation (UBQC) hides a client's computation by
using a computation-independent BFK09 brickwork graph and encoding the
computation in measurement angles, which limits the use of graph-changing
optimizations.  We study
blindness-preserving brickwork optimization (BPBO): certified local
resynthesis of BFK09-compatible brickwork patterns below the blinding layer.
BPBO detects one-, two-, and three-wire regions; for each candidate region it
either proves a semantic floor or supplies an executable witness, and it
accepts a replacement only after its branch-frame, output-frame, and blinding
behavior have been checked.  The optimized outputs remain standard brickwork
patterns and are evaluated with a logical qubit-recycled UBQC execution stack
that runs arbitrary-length patterns using $n\times 2$ active logical qubits.  The layer
evidence includes a one-wire H-count floor, a two-wire CNOT-cost floor, a
three-wire parity-ledger floor, a clean three-cell $\CCZ$ witness whose
optimality claim is scoped to the CNOT+T phase-gadget family, and an
endpoint-target three-cell CCX/Toffoli application witness; the fixed
middle-target CCX case is retained as a four-cell fallback.  The security
statement is a compatibility result: BPBO preserves UBQC blindness at the
declared optimized dimensions and remains compatible with inherited
verification guarantees under explicit test-round conditions, without
introducing a new trap-soundness theorem.  On Bell/CX, Grover-2,
endpoint-Toffoli, and Grover-3 evaluation cases, BPBO demonstrates certified
local reductions; in the largest case, Grover-3, the materialized pattern is
reduced from $3\times 725$ to $3\times 98$ while preserving the expected
marked-state statistics up to sampling noise.
\end{abstract}

\maketitle

\section{Introduction}\label{sec:intro}

A client who delegates a quantum computation to an untrusted server would
like the server to learn nothing about that computation---a goal predating
the protocols that achieve it~\cite{childs05}. Universal blind
quantum computation (UBQC) achieves this with a structural commitment: client
and server execute a measurement-based computation~\cite{rb01} on a brickwork
graph fixed \emph{independently} of the computation, so that all
computational content resides in a stream of one-time-padded measurement
angles~\cite{bfk09}; trap-based extensions make the delegation verifiable at
practical overhead~\cite{fk17,leichtle21}, and verifiable blind delegation
has reached hardware demonstration~\cite{drmota24,polacchi25}. The same
commitment blocks most direct pattern optimizations.
Many standard routes to a smaller \emph{pattern}---standardization
\cite{dkp07}, Pauli-flow preprocessing~\cite{simmons21}, and
flow-preserving or graph-level rewriting~\cite{mcelvanney22,graphix,duncan20}---reorganize
the command structure or change the graph/graph-like representation, while the
BFK09 leakage model requires the server-facing graph rule to remain
computation-independent. Circuit-level optimization~\cite{amm14} survives, but
only upstream of the lowering, where it cannot touch the redundancy the
lowering itself introduces.
In this model, the public topology rule constrains optimization, while the
declared dimensions remain permitted leakage. Pattern size is therefore an
optimizable resource only when every rewrite stays strictly below the blinding
layer, and provably so.

The strategy is therefore two-part.  Because savings from computation-dependent
server-facing graph changes are outside the admitted leakage surface unless
hidden by padding or a separate leakage argument, BPBO keeps each admitted
pattern inside the standard public brickwork family, declares the optimized
dimensions, and makes pattern size the certified object: a local shortening is
admitted only with a floor certificate or executable witness whose semantics,
frame behavior, and blinding compatibility are checked.  Because even an
optimized blind pattern may remain long (Grover-3 lowers to $3\times 725$
before optimization), a separate qubit-recycled runner executes the resulting
standard brickwork pattern column by column with $n\times 2$ active logical
qubits. The optimization side certifies what can be shortened without changing
the server-facing topology rule; the execution side makes the remaining long
blind pattern runnable.

We study blindness-preserving brickwork optimization (BPBO), an
arity-stratified certified region-resynthesis method for the standard BFK09
brickwork family at declared dimensions (Fig.~\ref{fig:framework}). BPBO
operates below the blinding layer. It detects local one-, two-, and three-wire
regions and produces two distinct forms of evidence: floor certificates, which
support lower-bound claims, and executable witnesses, which can be materialized
only after semantic, branch-frame, output-frame, and blinding-compatibility
checks. The accepted output is still a standard BFK09 brickwork pattern. A
separate logical qubit-recycled execution stack then runs the resulting pattern
at constant logical active width, using the two-column window naturally induced
by column-ordered brickwork measurements.

The name is meant literally.  \emph{Blindness-preserving} means that the
accepted executable rewrite changes only the declared pattern dimensions and
the hidden angle stream, not the public graph family or the UBQC transcript
interfaces.  \emph{Brickwork} fixes the domain to BFK09-compatible patterns.
\emph{Optimization} means local, certificate-gated shortening, not an
unconstrained circuit compiler.
Operationally, BPBO treats brickwork cells as resource suppliers and target
regions as requirement ledgers, then applies one certify--construct--admit
discipline across arities. The L1 layer handles one-wire regions through
H-count certificates and one-brick witnesses; L2 handles two-wire entangling
regions through CNOT-cost certificates and registered two-wire witnesses; L3
handles three-wire CCZ/CCX-class regions through parity-ledger certificates
and explicit application witnesses. A floor certificate is kept separate from
an executable replacement: the method may prove a lower bound even when the
submitted artifact runtime has no admitted materialization for that region.

The primary contribution is BPBO as a protocol-safe optimization contract,
supported by execution and evidence components.
\begin{itemize}
\item \emph{Blindness-preserving brickwork optimization.} We formalize BPBO
  as a certificate-gated local optimizer for BFK09-compatible brickwork patterns,
  with a common
  certify--construct--admit rule for L1 one-wire, L2 two-wire, and L3
  three-wire regions (Secs.~\ref{sec:bpbo}--\ref{sec:pipeline}).
\item \emph{Qubit-recycled execution.} We evaluate optimized brickwork
  patterns with a logical execution stack that uses $n\times 2$ active logical
  qubits, preserves the UBQC transcript interfaces under the admission
  conditions, and admits a dynamic-circuit lowering with mid-circuit
  measurement and reset~\cite{decross23}
  (Sec.~\ref{sec:platform}).
\item \emph{Protocol admission and evidence.} We prove branch-frame closure,
  a conditional UBQC-blindness compatibility theorem at the declared optimized
  leakage, and compatibility with inherited trap-based verification under
  explicit admission conditions; the implementation is
  checked on layer-representative cases, with Grover-3 used as an integrated
  application rather than as a broad compiler-performance benchmark
  (Secs.~\ref{sec:security} and~\ref{sec:results}).
\end{itemize}

The representative reductions are chosen to cover the layers rather than to
serve as an empirical optimality study. One-wire SYNTH1Q realizes an admitted
one-brick subset of the H-count floor; Bell/CX and Grover-2 exercise the
two-wire layer; endpoint-target CCX/Toffoli exercises an L3 application
witness; and Grover-3 combines four admitted L3 applications into a full UBQC
run. The benchmark corpus is implementation-facing, CCZ optimality is scoped
to the CNOT+T phase-gadget family, and verification compatibility is inherited
from the underlying protocol under the stated admission conditions. Together,
these cases exercise the L1, L2, L3, and integrated-composition admission paths
claimed by the submitted artifact.

The rest of the paper fixes the UBQC notation and leakage model
(Sec.~\ref{sec:prelim}), describes recycled brickwork execution
(Sec.~\ref{sec:platform}), defines BPBO and its layer certificates
(Secs.~\ref{sec:bpbo} and~\ref{sec:certificates}), states the admission and compatibility
conditions (Sec.~\ref{sec:security}), explains the implementation pipeline
(Sec.~\ref{sec:pipeline}), reports results (Sec.~\ref{sec:results}), and
then discusses related work, scope, and reproducibility.

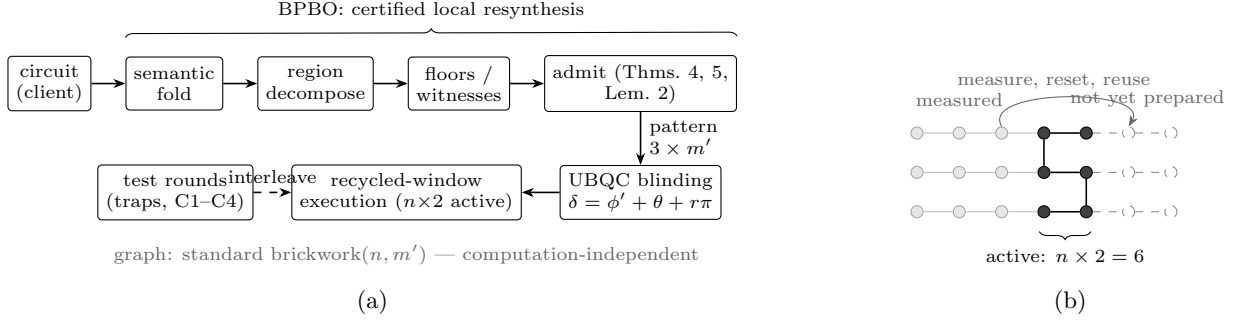
\begin{figure*}[tbp]
\centering
\begin{minipage}[b]{0.60\textwidth}
\centering
\begin{tikzpicture}[font=\scriptsize,
  box/.style={draw, rounded corners=1.5pt, align=center, inner sep=3pt,
              minimum height=7mm},
  flow/.style={-{Stealth[length=1.6mm]}, semithick}]
\node[box] (circ) {circuit\\(client)};
\node[box, right=4.5mm of circ] (fold) {semantic\\fold};
\node[box, right=4.5mm of fold] (dec) {region\\decompose};
\node[box, right=4.5mm of dec] (cert) {floors /\\witnesses};
\node[box, right=4.5mm of cert] (adm)
  {admit (Thms.~\ref{thm:p2}, \ref{thm:blind},\\Lem.~\ref{lem:v1})};
\draw[flow] (circ) -- (fold);
\draw[flow] (fold) -- (dec);
\draw[flow] (dec) -- (cert);
\draw[flow] (cert) -- (adm);
\draw[decorate, decoration={brace, amplitude=3pt}]
  ([yshift=2.5mm]fold.north west) -- ([yshift=2.5mm]adm.north east)
  node[midway, above=3.5pt] {BPBO: certified local resynthesis};
\node[box, below=7mm of adm] (blind)
  {UBQC blinding\\$\delta=\phi'+\theta+r\pi$};
\node[box, left=5mm of blind] (exec)
  {recycled-window\\execution ($n{\times}2$ active)};
\node[box, left=5mm of exec] (test) {test rounds\\(traps, C1--C4)};
\draw[flow] (adm) -- node[right, align=left] {pattern\\$3\times m'$} (blind);
\draw[flow] (blind) -- (exec);
\draw[flow, dashed] (test) -- node[above] {interleave} (exec);
\node[below=2mm of exec, align=center, text=black!60]
  {graph: standard brickwork$(n,m')$ --- computation-independent};
\end{tikzpicture}\\[1mm]
(a)
\end{minipage}\hfill
\begin{minipage}[b]{0.37\textwidth}
\centering
\begin{tikzpicture}[font=\scriptsize, x=5.6mm, y=5.2mm]
\foreach \c in {0,...,2}{\foreach \r in {0,...,2}{
  \node[circle, draw=black!35, fill=black!10, inner sep=1.6pt]
    (q\c\r) at (\c,-\r) {};}}
\foreach \c in {3,4}{\foreach \r in {0,...,2}{
  \node[circle, draw, fill=black!75, inner sep=1.6pt] (q\c\r) at (\c,-\r) {};}}
\foreach \c in {5,6}{\foreach \r in {0,...,2}{
  \node[circle, draw=black!45, dashed, inner sep=1.6pt] (q\c\r) at (\c,-\r) {};}}
\foreach \r in {0,...,2}{
  \draw[black!25] (q0\r) -- (q1\r) -- (q2\r);
  \draw[semithick] (q3\r) -- (q4\r);
  \draw[black!45, dashed] (q4\r) -- (q5\r) -- (q6\r);
  \draw[black!25] (q2\r) -- (q3\r);}
\draw[semithick] (q30) -- (q31);   
\draw[semithick] (q41) -- (q42);
\draw[{Stealth[length=1.5mm]}-, black!60, out=60, in=120, looseness=0.9]
  (q50.north) to node[above, pos=0.55] {measure, reset, reuse} (q20.north);
\draw[decorate, decoration={brace, mirror, amplitude=3pt}]
  ([yshift=-2mm]q32.south west) -- ([yshift=-2mm]q42.south east)
  node[midway, below=3.5pt] {active: $n\times 2=6$};
\node[above=1.5mm of q10, text=black!60] {measured};
\node[above=1.5mm of q50, xshift=2.5mm, text=black!60] {not yet prepared};
\end{tikzpicture}\\[1mm]
(b)
\end{minipage}
\caption{Blindness-preserving brickwork optimization (BPBO) and
qubit-recycled execution workflow. (a)~BPBO lowers a client-side basis stream
to fixed BFK09 brickwork regions, certifies candidate local replacements by
floors or witnesses, and admits only replacements whose frames and leakage
remain protocol-compatible. The materialized output is again a standard
brickwork pattern at the leaked dimensions $(n,m')$, with all computational
content in measurement angles. (b)~Recycled-window
execution: column-ordered measurement and nearest-column entanglement let two
live columns suffice---measured qubits are reset and their indices reused---so
a three-wire pattern of any length runs with six active qubits
(Sec.~\ref{sec:platform}).}
\label{fig:framework}
\end{figure*}

\section{Background and Security Model}\label{sec:prelim}
\subsection{Brickwork UBQC}
In the one-way model~\cite{rb01}, computation proceeds by single-qubit
measurements on an entangled resource state, with corrections propagated by
the measurement calculus~\cite{dkp07}. UBQC~\cite{bfk09} fixes the resource to
the \emph{brickwork} state $G_{n\times m}$: $nm$ qubits in
$\ket{+}$, arranged on $n$ horizontal wires and $m$ columns, joined by
CZ edges along each wire and by vertical rungs whose positions follow a fixed
period-$8$ stagger.  The graph depends only on the declared dimensions, not
on the computation.

The client prepares each qubit as $\ket{+_\theta}$ with
$\theta$ drawn uniformly from the BFK angle alphabet
\[
  \Abfk=\{k\pi/4 : k=0,\dots,7\}.
\]
The server entangles according to $G_{n\times m}$. For each qubit in column
order, the client announces
\[
  \delta=\phi'+\theta+r\pi,
\]
where $\phi'$ is the flow-adapted computation angle and $r$ is a fresh
uniform bit that flips the reported outcome's interpretation.  Measurement
$M^\delta$ projects onto
$\ket{\pm_\delta}=(\ket{0}\pm e^{i\delta}\ket{1})/\sqrt{2}$.
Pauli byproducts are tracked as a classical output frame rather than corrected
by the server.

\subsection{Server View and Leakage}
The server-visible protocol view consists of the declared dimensions, the
standard brickwork graph and column order, the received-qubit ensemble, the
$\delta$ stream, the server's measurement record, and the classical
interaction generated during execution.  Blindness in the sense
of~\cite{bfk09} states that, conditioned on the declared dimensions $(n,m)$,
this view is independent of the underlying computation.

BPBO changes only the accepted brickwork length. For an input circuit $C$, the
leakage surface after optimization is
\[
  L(C)=(n,m'),
\]
where $m'$ is the admitted optimized length. Thus the optimized dimensions are
intentional leakage: Theorem~\ref{thm:blind} is conditioned on $(n,m')$ rather
than on the unoptimized length.  A client that must hide the optimized size can
pad or bucket the pattern back to a canonical length, outside the optimizer;
the bucket policy must be fixed independently of the hidden computation within
the intended privacy class and may trade away some or all reported savings.
Optimizer logs, selected-region metadata, witness identifiers, timing records,
retry traces, and artifact handles are not server-visible protocol
messages unless an implementation explicitly sends them.

\subsection{Security Claim Scope}
The client is restricted to offline preparation of random single-qubit states
$\ket{+_\theta}$ and classical computation; it holds no quantum memory. The
server has unrestricted quantum capability and is arbitrarily malicious: no
computational assumption is made, and it may deviate from the protocol at
will. Communication consists of a one-way quantum channel from client to
server during preparation and a two-way classical channel during execution.

Within this model, the security role of BPBO is compatibility with the base
UBQC leakage model.  The admitted optimized pattern is executed at its
declared dimensions $(n,m')$, and Sec.~\ref{sec:security} states the
conditions under which the resulting view remains blind at that leakage
surface.  Verification guarantees are inherited from the underlying
trap-based protocol~\cite{fk17,leichtle21} when the optimized pattern satisfies
the compatibility conditions specified in Sec.~\ref{sec:security}; BPBO does
not introduce a new trap-soundness theorem.  Qubit recycling is an execution
schedule for the standard brickwork pattern after materialization; it does not
alter the prepare-and-send client model or make allocation metadata part of the
server-visible computation beyond the declared dimensions.

Equally explicit are the \emph{non-claims}: (i)~no new verification-soundness
theorem is proved---soundness is inherited, never re-derived; (ii)~no
composable-security analysis is attempted---claims live in the standalone
model of~\cite{bfk09}; (iii)~the empirical evidence of
Sec.~\ref{sec:results} is simulation-level, not a hardware demonstration; and
(iv)~pattern dimensions remain leaked by definition, so size hiding requires
padding or bucketing that may forgo some optimization savings. These
boundaries are revisited where each is load-bearing
(Secs.~\ref{sec:security} and~\ref{sec:discussion}).

\section{Qubit-Recycled Execution of Brickwork Patterns}\label{sec:platform}

Prior circuit-based UBQC/MBQC simulation demonstrated the feasibility of
small blind computations, including a two-qubit Grover instance, on
gate-model platforms~\cite{leechung25}. This section takes that execution
viewpoint as its starting point but separates it from the optimizer: BPBO
shortens admitted standard brickwork patterns, while the runner executes and
tests those patterns through a two-column active window.

A brickwork pattern of $n$ rows and $m$ columns has $nm$ qubits; the raw
Grover-3 lowering already needs $3\times 725 = 2175$, far beyond full
statevector simulation. The runner's enabling observation is structural:
all entanglers are mutually commuting CZs, horizontal edges connect only
adjacent columns, and measurements are column-ordered---so a qubit's life is
short. The \emph{streaming-window runner} keeps only two columns alive:
before measuring column $c$, it prepares column $c{+}1$ (fresh $\ket{+}$
qubits, the column's vertical rungs, and the horizontal edges into it),
measures column $c$ destructively, removes those qubits from the state, and
reuses their indices. Peak active qubits are $n\times 2$---six for every
three-wire pattern in this paper, and eight in a four-wire smoke test
included in the artifact package---independent of $m$; memory falls from
$2^{nm}$ amplitudes to $2^{2n}$, and time is linear in $m$
(Table~\ref{tab:resources}).

\begin{table*}[t]
\caption{Hardware-facing resource accounting for representative accepted
geometries.  Measurements/resets count logical mid-circuit measurement events
and reset opportunities; a backend may omit the final physical reset.  Column
cycles are adaptive measurement/feed-forward rounds, equal to columns minus
one because the output column is not measured.}
\label{tab:resources}
\centering\scriptsize
\setlength{\tabcolsep}{3pt}
\resizebox{0.82\textwidth}{!}{%
\begin{tabular}{lccccc}
\toprule
case & geometry & full amplitudes & active & meas./reset & cycles \\
\midrule
Bell/CX & $2{\times}5$ & $2^{10}$ & $4$ & $8$ & $4$ \\
Grover-2 & $2{\times}29$ & $2^{58}$ & $4$ & $56$ & $28$ \\
clean $\CCZ$ patch & $3{\times}25$ & $2^{75}$ & $6$ & $72$ & $24$ \\
endpoint CCX & $3{\times}33$ & $2^{99}$ & $6$ & $96$ & $32$ \\
Grover-3 raw reference & $3{\times}725$ & $2^{2175}$ & $6$ & $2172$ & $724$ \\
Grover-3 optimized & $3{\times}98$ & $2^{294}$ & $6$ & $291$ & $97$ \\
$n{=}4$ smoke & $4{\times}24$ & $2^{96}$ & $8$ & --- & $23$ \\
\bottomrule
\end{tabular}}
\end{table*}

\emph{Exactness.} Just-in-time preparation is equivalent to whole-graph
preparation: a measurement on qubit $q$ commutes with every future CZ not
acting on $q$; every CZ acting on $q$ is applied before $q$ is measured
(window $\ge 2$ guarantees this for nearest-column edges); and the adaptive
corrections depend only on classical outcomes, not on when future qubits were
prepared. This equivalence is machine-checked at path-graph, elementary-cell,
and long three-wire scales; the detailed replay thresholds and branch-count
artifacts are reported with the implementation and reproducibility data.  The
point here is structural: recycling changes allocation time, not the
brickwork measurement pattern.

\emph{Adaptive corrections.} Measurement angles adapt by the standard MBQC
rule with dependencies generated east-flow on the brickwork---the brickwork's
causal flow~\cite{bkmp07} is strictly column-ordered, which is also why a
two-column window is the canonical streaming unit rather than a tuned cache
size; every branch of
the $H/T/\mathrm{CNOT}$ elementary cells reproduces the zero branch modulo an
output Pauli frame. The runner uses the same byproduct-frame algebra later
used by BPBO admission, so the runner and admission checks share a single
classical tracking convention~\cite{dkp07}.

\emph{From simulation to dynamic circuits, and the UBQC interfaces.} The window
runner lowers to a logical dynamic-circuit schedule with
$n\times\text{window}$ active qubits, mid-circuit measurement and reset, and
the adaptive corrections as classical conditionals---the same
measure-early-and-reuse principle
established for circuit-model compilation by qubit-reuse
techniques~\cite{decross23}; the brickwork's column structure makes the
two-column reuse schedule canonical for this logical dependency structure
rather than a search problem.  One column cycle is:
prepare the recycled next-column qubits in their blinded $\ket{+_\theta}$
states; apply the fixed vertical-rung CZ layer and the horizontal inter-column
CZ layer; measure the retiring column in the client-provided adaptive bases;
report outcomes and update the branch frame/output decoder; reset the retired
column indices for reuse.  Thus a pattern of $m$ columns has $m{-}1$ adaptive
cycles and time
\[
T_{\mathrm{pattern}} \approx (m-1)\,\tau_{\mathrm{col}} + T_{\mathrm{out}},
\]
where $\tau_{\mathrm{col}}$ is the backend-specific
prepare/CZ/measure/feed-forward/reset latency.  We report this as a logical
schedule and do not claim a backend hardware trace. Statistical validation
and equivalence gates are reported with the implementation and results
sections. The UBQC transcript interfaces are preserved: blinded angle streams,
client/server transcript separation, and input/output one-time-pad
decoding---recycling changes \emph{when} logical qubits exist, not the
server-visible data associated with each qubit.

\section{BPBO Method Contract}\label{sec:bpbo}

BPBO (blindness-preserving brickwork optimization) is a certifying
\emph{local} optimizer for fixed BFK09 brickwork patterns.  Its generality is
an input-domain statement: BPBO can analyze an arbitrary BFK09-compatible
pattern or basis stream.  It is not a claim of global optimality, complete
synthesis, or an executable replacement for every certified floor.  The
contract is:
\begin{center}
\begin{minipage}{0.94\columnwidth}\small
\textbf{Input:} a fixed BFK09 brickwork pattern, or a basis stream that
materializes to one.\\
\textbf{Output:} either the original pattern, or a shorter standard BFK09
pattern $P'$ plus an execution-plan certificate.\\
\textbf{Guarantee:} if BPBO outputs $P'$, then each admitted replacement is
semantically equivalent to the original region up to tracked frames, remains
within the BFK09 angle alphabet, executes with the declared leakage surface
$(n,m')$, and has passed branch-frame, output-frame, and materialized-saving
checks.\\
\textbf{Non-guarantee:} BPBO does not claim a global optimum, a complete
witness compiler, or executable replacements for all floor-certified regions.
\end{minipage}
\end{center}

\subsection{Resource and Requirement View}
BPBO rests on two changes of viewpoint. A brickwork cell is not ``a gate
slot'' but a \emph{resource supplier}: each wire's chain through one cell
supplies exactly two effective Hadamards (the chain factors as
$R_z\,H\,R_z\,H\,R_z$---this single-wire restriction is the $n{=}1$ cell of
$\Rbfk{1}{k}$; larger macro-cells scale these supplies with their
columns), the vertical rungs supply entangling capacity, and
every measured column supplies a free phase from $\Abfk$. Dually, a target
unitary is not ``a circuit'' but a \emph{requirement ledger}: how many
Hadamards, how much entangling capacity, and---per
Lemma~\ref{lem:ledger}---which parity forms must receive odd phase deposits.
Optimization is then a matching problem between supply and requirements, and
it can be \emph{certified} at both ends---by a single principle on the
lower side (Theorem~\ref{thm:floor}), and by an explicit angle-table
witness on the upper.

\subsection{Common Certify--Construct--Admit Rule}
Every accepted rewrite follows the same local rule:
\begin{enumerate}
\item parse the input into a BFK09-compatible stream and detect a candidate
region;
\item canonicalize the region semantics modulo gauges and Pauli frames;
\item compute the region certificate or floor;
\item construct or retrieve a BFK09-realizable witness;
\item verify zero-branch semantic equivalence;
\item check branch-frame closure, output-frame admissibility, and positive
materialized saving;
\item materialize the optimized standard BFK09 pattern, emit its execution
plan, and validate the accepted pattern.
\end{enumerate}
A floor certificate is a lower-bound statement, not an executable
replacement.  Admission is a separate runtime gate.  The status ladder is
\begin{center}
\footnotesize
\begin{tabular}{c}
\textsc{Detected}$\to$\textsc{Canonicalized}$\to$\textsc{Floor}$\to$\textsc{Witness}\\
$\to$\textsc{Admitted}$\to$\textsc{Materialized}$\to$\textsc{Validated}.
\end{tabular}
\end{center}

\subsection{Certified Resynthesis Layers}\label{sec:layers}
BPBO applies this discipline at three resynthesis arities.  The layer names
are not new primitive gates; they are the public organization of the
optimizer's certified region families.
\begin{center}
\scriptsize
\begin{tabular}{@{}llll@{}}
\toprule
Layer & Scope & Floor & Gate \\
\midrule
L1 & 1-wire & H-count $\lceil h/2\rceil$ & R2/R9/R10 \\
L2 & 2-wire & CNOT-cost $\lceil c/2\rceil$ & E1-T/R12/R11/L2 \\
L3 & 3-wire & parity ledger & CCZ/CCX \\
\bottomrule
\end{tabular}
\end{center}
This table is also a scope statement.  The floors are mathematical lower
bounds on region classes; the runtime executes only the admitted subsets for
which a BFK09 witness, branch-frame behavior, output-frame handling, and a
positive materialized saving have all been checked.  Thus a region may be
floor-certified or synthesis-available without being selected for execution.
Surviving cells are handled later by materialization and scheduling; they are
not a fourth resynthesis arity.  The optimized payload is always the
materialized BFK09 pattern, not the preview log that discovered it.

\section{BPBO Layer Certificates and Witnesses}\label{sec:certificates}
The previous section defined BPBO as a local certify--construct--admit
method.  This section supplies the certificates and witnesses used by that
method: the reachable-family notation, the common lower-bound theorem, its
one-, two-, and three-wire instantiations, and the L3 primitive and application
witnesses that connect the floor certificates to evaluated regions.
The claims in this section climb a fixed ladder.  A \emph{semantic target} is
a zero-branch unitary modulo the stated gauge and output-frame conventions; a
\emph{floor certificate} is a lower bound in a stated realization class; a
\emph{schedule hint} is a feasible assignment in an over-approximating menu; an
\emph{explicit witness} is an angle table realizing the target modulo frame;
and a \emph{runtime-admitted witness} additionally passes the branch-frame,
output-frame, and materialized-saving gates of Sec.~\ref{sec:security}.  Exact
cell complexity is claimed only where a scoped lower bound and an explicit
witness meet in the same realization class.

\subsection{Reachable Families}
Let $\Rbfk{n}{k}$ denote the unitaries realizable by $k$ consecutive clean
cells on $n$ wires modulo the permitted boundary gauge and a single-qubit
Pauli output frame, at the zero-branch level (runtime admission is
Sec.~\ref{sec:security}).  For L1/L2 a clean cell is the corresponding
one- or two-wire brickwork unit.  For L3, the counted unit is the clean
three-wire macro-cell defined in Sec.~\ref{sec:ccz}; the reported floor
battery uses the same start convention as the registered witness artifacts.
Its Clifford quotient is closed and fast: modulo Pauli frames, Clifford
reachability uses the binary symplectic/tableau representation of Clifford
action \cite{aaronson04,gottesman98}, reducing to
$\mathrm{Sp}(2n,\mathbb{F}_2)$, of order $6$, $720$, and $1\,451\,520$ for
$n=1,2,3$.  The implementation also keeps signed Clifford tableaus, with
$24$, $11\,520$, and $92\,897\,280$ representatives at these arities.  For the
fixed local arities used here, membership and composition are table-backed
operations (the layer is exhaustively enumerated for $n\le 2$ and checked
against sampled $n=3$ consistency tests). The non-Clifford certificate families
below use the deposit ledger---the phase-polynomial representation of CNOT+T
and phase-gadget circuits~\cite{amm14,cowtan20}, transplanted to a setting
where supply is a fixed cell menu rather than freely placed gates
(Sec.~\ref{sec:related}).

\subsection{The Supply--Demand Floor}\label{sec:sdfloor}
All formal BPBO floor certificates used below fit one bookkeeping template.

\begin{definition}[requirement map; supply filtration]\label{def:sdf}
Fix an arity $n$ and a realization class $F$, including its cell or macro-cell
unit, zero-branch convention, boundary gauges, and output-frame quotient.  A
\emph{requirement map} assigns to each target class $[U]_F$ a demand $D(U)$ in
a partially ordered set, constant on the allowed gauge--Pauli orbit. A
\emph{supply filtration} for $F$ is a monotone family
$S(0)\subseteq S(1)\subseteq S(2)\subseteq\cdots$ that over-approximates
reachability: $U\in\mathcal{R}_F(n,k)$ implies $D(U)\in S(k)$.
\end{definition}

\begin{theorem}[Supply--Demand Floor]\label{thm:floor}
For any requirement map $D$ and supply filtration $S$
(Definition~\ref{def:sdf}) over-approximating the class $F$,
\[
  k^{F}_{\min}(U)\;\ge\;\mathrm{floor}_D(U):=
  \min\{\,k\ge 0 : D(U)\in S(k)\,\},
\]
with the convention that the minimum is $+\infty$ if the set is empty.
\end{theorem}
The proof is one paragraph (App.~\ref{app:proofs}): a $k$-cell realization
in $F$ would place $D(U)$ in $S(k)$ by over-approximation, with
orbit-invariance discharging the realization's gauge and frame dressing.
Two structural consequences carry through every instantiation. First, the
bound is relative to its class $F$: where $S$ over-approximates all clean
windows the floor is unconditional, and where it over-approximates the
CNOT+T phase-gadget family the floor is family-scoped---each corollary
below carries its scope explicitly. Second, over-approximation means
floor-tightness is \emph{not} promised: $D(U)\in S(k)$ does not produce an
angle-level witness at $k$, and the executability ladder of
Sec.~\ref{sec:pipeline} tracks this distinction operationally; the endpoint/middle
CCX layout gap in Sec.~\ref{sec:results} exhibits it live.

\noindent\emph{L1 certificate.}
\begin{corollary}[Hadamard grading, $n=1$]\label{thm:hcount}
Within the one-wire clean-cell model and fixed boundary/frame convention, let
$D(U)=h(U)$ be the minimal Hadamard count over $\Abfk$-decompositions.  With
$S(0)=\{h=0\}$ and $S(k)=\{h\le 2k\}$ for $k\ge 1$ (each cell's wire chain
supplies exactly two Hadamards), the floor is $\lceil h(U)/2\rceil$. At
$n=1$ the filtration is exact---$\Rbfk{1}{k}=\{U : h(U)\le 2k\}$---so the
floor is the one-wire clean-cell cost.
\end{corollary}
The cell is Hadamard-limited, not $T$-limited: $\Abfk$ phases ride in the
available angle slots (proof sketch in App.~\ref{app:proofs}; exact
cyclotomic enumeration confirms the small-$k$ filtration used by the
artifact). This inverts the resource intuition inherited from $T$-count-centric
circuit optimization.
Operationally, this is the L1 one-wire certified resynthesis layer.  The
production implementation uses deterministic $H;H\mapsto I$ cancellation,
direct one-wire template replacement, and the runtime-admitted one-brick
family displayed as SYNTH1Q.  Only replacements with exact semantic
equivalence, branch-frame admission where required, output-frame admissibility,
and positive materialized saving enter the materialized UBQC pattern; the
theorem above is broader than the current executable subset.

\noindent\emph{L2 certificate.}
\begin{corollary}[Entangling grading, $n=2$]\label{thm:l2}
With $D(U)=c(U)$, the local-equivalence CNOT cost identified from Makhlin
invariants~\cite{makhlin02} and entangling-gate lower bounds~\cite{sbm06}, and
$S(0)=\{c=0\}$, $S(k)=\{c\le 2k\}$ for $k\ge 1$ (two rungs per cell),
$\lceil c(U)/2\rceil$ is an entangling-cell floor.  Equality is certified only
for the named finite classes used below: the Clifford pre-contexts and the
finite $\{I,T,T^\dagger\}$ CNOT-context neighborhoods (exhaustive
verification; App.~\ref{app:proofs}).
\end{corollary}
Operationally, this is the L2 two-wire certified resynthesis layer.  The
runtime attempts only registered L2 candidates: T-context (\texttt{E1-T}),
synthesized-context (\texttt{R12-E-pre}), Clifford-context (\texttt{R11}), and
short-region reduction (\texttt{L2-Reduce}).  Candidates are materialized only
after exact zero-branch semantics, branch-frame closure or an equivalent
witness, output-frame discharge, and positive materialized saving pass.  Thus
the claim is not that every two-wire unitary is executed at $\lceil c/2\rceil$,
nor that local one-wire work has zero cost; it is that the entangling-capacity
lower bound is general within the stated cell model, and floor-tight rewrites
are supplied for the finite registered classes.

\subsection{L3 Certificate: The Three-Wire Floor Algorithm}\label{sec:floor}
At $n=3$ the same instantiation pattern yields not a formula but an
algorithm.  The target universe here is the clean three-wire CNOT+T
phase-gadget family after canonicalization over per-wire Hadamard boundary
gauges and Pauli output frames.  The demand $D(U)$ is the resulting affine
linear skeleton together with the nonconstant odd-support syndrome of
Lemma~\ref{lem:ledger}.  For $k$ clean macro-cells, $S(k)$ is the generous
$k$-cell schedule menu: the alternating $2k$ two-rung blocks may realize any
allowed $\mathrm{GL}(2,\mathbb{F}_2)$ action on their active wire pair, and
odd deposits may be placed on parity forms carried by the wires under the
stated start convention.  Lemma~\ref{lem:coverage} is the two-cell instance
used for the $\CCZ$ no-go below; the floor algorithm uses the same menu for
general $k$.

\begin{corollary}[Parity coverage, $n=3$]\label{cor:n3}
With this demand and supply menu, $\mathrm{floor}_D$ is a sound lower bound on
clean-macro-cell cost within the stated phase-gadget family, computable by
finite search.
\end{corollary}

\noindent The production implementation certifies the floor, not necessarily
an executable replacement, for an arbitrary in-scope target:
\begin{enumerate}
\item \emph{Canonicalize}: search the $64$ per-wire Hadamard boundary-gauge
pairs for $G\,U\,G'$ in affine-monomial form (this automatically reduces
Toffoli and the Hadamard-dressed Grover blocks to $\CCZ$); report
out-of-scope honestly if none exists.
\item \emph{Extract requirements}, both Pauli-frame invariant: the linear
skeleton (which input parity each output wire carries) and the nonconstant
odd-support syndrome of Lemma~\ref{lem:ledger}.
\item \emph{Search the schedules}: for $k=0,1,2,\dots$, enumerate the generous
$k$-cell menu for an assignment realizing the skeleton while exposing every
parity in the syndrome. The first feasible $k$ is the certified floor, and the
feasible assignment is emitted as a synthesis hint rather than as an angle-level
witness.
\end{enumerate}
The public artifact reports the positive-cost floor battery using the clean
START=5 macro-cell convention up to its stated finite cap; cap failures are
reported as unresolved, not as no-go theorems.  Soundness is
Theorem~\ref{thm:floor}'s over-approximation direction: within the stated clean
phase-gadget family and macro-cell convention, any physical gadget realization
would appear in the generous menu, so infeasibility in that menu proves a lower
bound for that family. On the \texttt{FLOOR-3W} battery
$\{\mathrm{CNOT},\ \mathrm{CX}_{01}\mathrm{CX}_{21},\ tc,\ \CCZ,\
\mathrm{CCX}_{\mathrm{mid}},\ H^{\otimes 3}\CCZ\}$---where $tc$ denotes the seven-gate
CNOT${}+T$ Toffoli phase core---Corollary~\ref{cor:n3}'s search returns floors
$(1,1,2,3,3,3)$, matching the registered battery values. These are lower
bounds until a target-specific witness achieves them; Sec.~\ref{sec:ccz}
supplies that meeting point for the central $\CCZ$ witness and the registered
application witnesses.

\subsection{L3 Case Study: Phase-Gadget CCZ Witnesses and Optimality}\label{sec:ccz}

The three-wire BPBO layer uses parity-ledger certificates and registered
witnesses for the evaluated $\CCZ/\mathrm{CCX}$-class regions.  $\CCZ$ is the
representative case because the registered Toffoli and Grover-3 application
regions reduce to $\CCZ$ up to local layers. This case study gives the exact
clean-macro-cell value for $\CCZ$ within the certified CNOT+T phase-gadget
family~\cite{cowtan20} and an unconditional three-cell upper bound in the full brickwork model;
it does not prove a family-free two-cell no-go.

\subsubsection{Setting}
A \emph{macro-cell} is a nine-column window of the three-wire brickwork; its
internal rung schedule is fixed by the window's start column modulo $8$, and the
two \emph{clean} phases (start $\equiv 5,7 \bmod 8$) carry no boundary rung.
A \emph{$k$-cell realization} of a unitary $U$ assigns angles from $\Abfk$ to
$k$ consecutive clean macro-cells whose zero-branch map equals $P\,U$ for some
single-qubit Pauli tensor $P$ (the output \emph{frame}; admitted witnesses pass
the output-frame checks of Sec.~\ref{sec:security}). Within this case study, the
certified scope is the \emph{phase-gadget family}: realizations that factor as
CNOT+T gadgets---monomial blocks at the rungs with diagonal phase deposits on
the parity forms carried by the wires. The registered production witnesses
reported here lie in this family; realizations outside it are addressed by
adversarial search (Remark~\ref{rem:scope}).
Column counts use four conventions.  A macro-cell window has nine local
columns including its boundary; each macro-cell contributes eight
angle-carrying measured columns; connected patch columns merge shared
boundaries, so three start-$5$ macro-cells form a $3\times25$ connected patch
with $72$ measured qubits; and period columns count the $24$-column alignment
period used when Grover blocks are concatenated.

\subsubsection{The requirement and the obstruction}
By the parity-ledger analysis of Sec.~\ref{sec:certificates}, every phase-gadget
realization of $\CCZ$ modulo a single-qubit Pauli frame must deposit odd
(T-family) phases on \emph{all seven} parity forms---the deposit syndrome is a
representation-independent invariant (Lemma~\ref{lem:ledger},
App.~\ref{app:proofs})---and in particular on the three $x_0$-containing
parities $x_0{\oplus}x_1$, $x_0{\oplus}x_2$, and
$x_0{\oplus}x_1{\oplus}x_2$, at moments when those forms are physically
carried by a wire. Whether a window can \emph{schedule} those deposits is a
finite question about its fixed rung alternation, and for two cells the
answer is negative:

\begin{theorem}[Phase-gadget two-cell no-go]\label{thm:nogo}
Within the phase-gadget family, no clean two-macro-cell window (start
$\equiv 5$ or $7 \bmod 8$) realizes $\CCZ$ modulo a single-qubit Pauli output
frame.
\end{theorem}

The mechanism is geometric (Fig.~\ref{fig:obstruction}): the wire $x_0$ mixes
only at $(0,1)$-blocks, and a two-cell window alternates blocks too coarsely to
expose the outer pair $x_0{\oplus}x_2$ together with the remaining required
parities and still return the wires to identity. The proof is a finite
enumeration ($6^4$ and $3\cdot 6^3\cdot 3$ assignments) over an
\emph{over-approximated} menu, which is what makes the bound sound
(App.~\ref{app:proofs}).

\begin{remark}[Scope]\label{rem:scope}
Outside the gadget family---arbitrary-angle windows whose internal dressing
does not factor as CNOT+T gadgets---the statement rests on adversarial search:
approximately $1.8\times 10^{5}$ randomized and structured two-cell candidate windows
cap at frame fidelity $0.730$ (start $5$) and $0.854$ (start $7$) against
$\CCZ$. We state Theorem~\ref{thm:nogo} as a \emph{bounded} no-go and do not
claim a family-free impossibility.
\end{remark}

\begin{theorem}[Three-cell witness]\label{thm:witness}
There exist angles in $\Abfk$ on three consecutive clean start-$5$ macro-cells
whose composite equals $(Y{\otimes}X{\otimes}Z)\cdot\CCZ$ exactly. The witness
(schedule $(\mathrm{rotA},\mathrm{CXb})^{\times 3}$: in each cell the
$(1,2)$ block acts by local rotations and the $(0,1)$ block carries the
CNOT; full angle tables in
App.~\ref{app:witness}) is verified \emph{exactly}: with all angles integer
multiples of $\pi/4$, the composite cell map lies in
$\mathbb{Z}[\zeta_8]^{8\times 8}$, and its proportionality to
$(Y{\otimes}X{\otimes}Z)\cdot\CCZ$ is a division-free cross-multiplication
identity in cyclotomic integer arithmetic---no floating point enters the
certificate. Independent floating-point reconstruction agrees (elementwise
deviation below $10^{-15}$), and the witness was
cross-verified by a second, independently implemented toolchain. The witness is
itself a phase-gadget realization---CNOT+T blocks with boundary and diagonal
deposits---in the same family as Theorem~\ref{thm:nogo}, so the lower and upper
bounds meet within a single realization class.
\end{theorem}

\begin{corollary}\label{cor:three}
Within the CNOT+T phase-gadget family, the clean-macro-cell complexity of
$\CCZ$ on the three-wire brickwork is exactly three
(Theorems~\ref{thm:nogo} and~\ref{thm:witness}). For unrestricted-angle
windows we prove the three-cell \emph{upper} bound (the witness is
unconditional) and report adversarial evidence against two cells
(Remark~\ref{rem:scope}); we do not claim a family-free lower bound.
\end{corollary}

\subsubsection{How the witness is constructed: two synthesis principles}
The witness is not a lucky search result; it is produced by a reusable
procedure built on two principles.

\begin{lemma}[Frame-chained synthesis]\label{lem:chaining}
Let $T_1,\dots,T_k$ be per-cell targets with $T_k\cdots T_1 = U$. Solve the
cells sequentially against \emph{residue-adapted} targets: cell $j$ realizes
$P_j\,G_j\,T_j\,R_{j-1}^{\dagger}$ exactly, where $P_j$ is a single-qubit
Pauli output frame,
$R_{j-1}=(U_{j-1}\cdots U_1)(T_{j-1}\cdots T_1)^{\dagger}$ is the exact
accumulated residue, $G_j$ a per-wire Hadamard boundary gauge, and $G_k$
trivial. Then $U_k\cdots U_1 = P\,U$ with $P$ the final cell's Pauli frame, by
construction. The gauge masks $G_j$ are a search dimension: greedy choices can
dead-end, and existence of a closing assignment is established per schedule.
\end{lemma}

Chaining is the synthesis-time instance of the byproduct-adaptation rule proved
in Sec.~\ref{sec:security}: each cell's residue is absorbed into the next
cell's target exactly as runtime byproducts are absorbed into future
measurement angles. Without it, individually perfect cells compose to fidelity
$\cos(\pi/4)\approx 0.707$---the frames of earlier cells flip the phase
deposits of later ones.

The second principle is a scheduling rule: a cell's trailing boundary always
emits per-wire Hadamard gauge that the \emph{next} cell absorbs, so a block
whose linear action is iSWAP-class (a same-pair double-CNOT ``rotation'')
cannot occupy the final cell, whose outgoing gauge must be trivial. Re-running
the schedule enumeration under this \emph{end-safe} constraint yields four
candidate schedules, all of which then solve at fidelity $1.0$; the rule's
mechanism is established by a Makhlin-invariant scan~\cite{makhlin02} of the two-rung block
family, and we use it as a verified design rule rather than a closed-form
theorem.

The same machinery yields, beyond bare $\CCZ$: the Grover block
$B=H^{\otimes 3}\,\CCZ$ as a three-cell witness (the diffusion Hadamard layer
absorbed into the boundary gauge), and, by composing four chained blocks, a
complete twelve-cell Grover-3 reference whose physical correctness is measured
in Sec.~\ref{sec:results}. Because each block spans $24\equiv 0 \pmod 8$
columns, consecutive blocks preserve the stagger phase and concatenate with
zero alignment padding.
The endpoint-target $\mathrm{CCX}$/Toffoli witness is the other runtime-facing
L3 application: it closes at the three-cell floor for target wire $x_2$, while
the fixed middle-target placement remains a registered four-cell fallback
because the three-row brickwork geometry is not row-symmetric.  Sec.~\ref{sec:results}
and App.~\ref{app:witness} separate these two application witnesses explicitly.

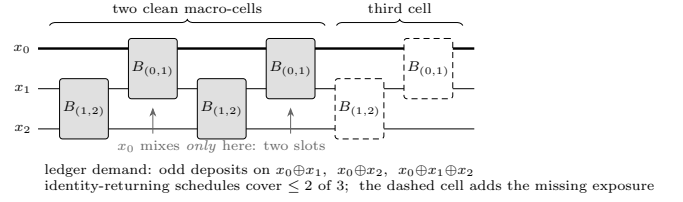
\begin{figure}[t]
\centering
\resizebox{\columnwidth}{!}{%
\begin{tikzpicture}[font=\scriptsize, x=1mm, y=1mm,
  blk/.style={draw, rounded corners=1.5pt, fill=black!12, inner sep=0pt,
              minimum width=8.5mm, minimum height=10.5mm},
  blkz/.style={blk, fill=white, densely dashed}]
\draw[very thick] (0,0)   node[left] {$x_0$} -- (76,0);
\draw (0,-7)  node[left] {$x_1$} -- (76,-7);
\draw (0,-14) node[left] {$x_2$} -- (76,-14);
\node[blk] (b1) at (8,-10.5)  {$B_{(1,2)}$};
\node[blk] (b2) at (20,-3.5)  {$B_{(0,1)}$};
\node[blk] (b3) at (32,-10.5) {$B_{(1,2)}$};
\node[blk] (b4) at (44,-3.5)  {$B_{(0,1)}$};
\node[blkz] (b5) at (56,-10.5) {$B_{(1,2)}$};
\node[blkz] (b6) at (68,-3.5)  {$B_{(0,1)}$};
\draw[decorate, decoration={brace, amplitude=3pt}]
  (2,4) -- (50,4) node[midway, above=3pt] {two clean macro-cells};
\draw[decorate, decoration={brace, amplitude=3pt}]
  (52,4) -- (74,4) node[midway, above=3pt] {third cell};
\draw[{Stealth[length=1.5mm]}-, black!60] (b2.south) ++(0,-1.5)
  -- ++(0,-4.5) coordinate (m1);
\draw[{Stealth[length=1.5mm]}-, black!60] (b4.south) ++(0,-1.5)
  -- ++(0,-4.5) coordinate (m2);
\node[black!60, below=0mm of m1, xshift=12mm]
  {$x_0$ mixes \emph{only} here: two slots};
\node[align=left, anchor=north west] at (0,-19)
  {ledger demand: odd deposits on $x_0{\oplus}x_1$,\;
   $x_0{\oplus}x_2$,\; $x_0{\oplus}x_1{\oplus}x_2$\\
   identity-returning schedules cover $\le 2$ of $3$;\;
   the dashed cell adds the missing exposure};
\end{tikzpicture}}
\caption{L3 phase-gadget obstruction for two clean macro-cells. A clean
two-macro-cell window alternates two-rung blocks on wire pairs $(1,2)$ and
$(0,1)$ in a fixed order set by the stagger phase. Tracking which parity forms
each wire carries through the block sequence, the three $x_0$-containing
parities
required by $\CCZ$'s ledger can never all visit a wire while the schedule
returns the wires to identity: every identity-returning assignment covers at
most two of $\{x_0{\oplus}x_1,\,x_0{\oplus}x_2,\,x_0{\oplus}x_1{\oplus}x_2\}$.
A third cell adds the missing exposure, and Theorem~\ref{thm:witness} provides
the corresponding L3 witness.}
\label{fig:obstruction}
\end{figure}

\section{Admission, Blindness, and Verification Compatibility}\label{sec:security}

A rewrite certified at the zero-branch level (Sec.~\ref{sec:ccz}) is not yet a
protocol participant. Four gaps remain between a frame/gauge-correct unitary
and an admissible pattern: measurement outcomes are random, the full output
decoder must compose all Pauli-frame sources, the angles are subsequently
blinded, and the protocol may interleave verification rounds. This section
states the protocol-level admission conditions that close these gaps---promoting
a witness from ``correct'' to ``admissible in the running protocol'' without
introducing a new security model.

\begin{theorem}[Branch-frame closure]\label{thm:p2}
Consider a clean pattern of $N$ measured columns with unblinded base angles
$\alpha_{r,c}\in\Abfk$.  The client applies the deterministic sign adaptation
$(-1)^{x_r}\alpha_{r,c}$, while the tracker bits $(x_r,z_r)$ per wire evolve,
per column, by (i)~$z_r \mathrel{\oplus}= s_{r,c}$ for outcome $s_{r,c}$,
(ii)~$z_a \mathrel{\oplus}= x_b,\; z_b \mathrel{\oplus}= x_a$ for each rung
$(a,b)$ in the column, and
(iii)~$(x_r,z_r)\leftarrow(z_r,x_r)$ at the column hop. Then for
\emph{every} outcome string $s$, the branch map equals $P_{\rm br}(s)\,U_0$
up to global phase, where $U_0$ is the zero-branch map and
$P_{\rm br}(s)=\prod_r X^{x_r}Z^{z_r}$ is the tracker's final Pauli. Moreover
$P_{\rm br}(s)$ is client-computable from $s$ and the public schedule, and all
adapted base angles remain in $\Abfk$.
\end{theorem}

Theorem~\ref{thm:p2} is stated so as to be its own implementation
specification: the tracker's update rules \emph{are} the runtime's classical
feed-forward, and its final Pauli \emph{is} the decoder's relabeling mask.
For a fixed branch $s$, the measurement projector is
$\bra{+_{(-1)^{x_r}\alpha_{r,c}+\pi s_{r,c}}}$; the $+\pi s$ term is branch
\emph{semantics} (the outcome-$1$ projector), not an online instruction.  This
is the standard MBQC feed-forward in an equivalent byproduct gauge: the usual
$(-1)^x\alpha+z\pi$ convention folds the tracked $z$-shift into the
transmitted angle, whereas we keep it in the classical frame.  The two gauges
produce identically distributed blinded transcripts and differ only by output
frame bookkeeping.
The same algebra serves twice---at compile time, the frame-chained synthesis
of Lemma~\ref{lem:chaining} absorbs the deterministic inter-cell residues
into the angle tables; at run time, the identical rules absorb the
measurement randomness. The implementation validates this branch closure as an
artifact gate (Sec.~\ref{sec:results} and App.~\ref{app:repro}); those
validation counts are evidence for the implementation, not an additional
security theorem. The decoder must accordingly compose every output-frame
term,
\[
F_{\rm out}(s)=F_{\rm br}(s)\oplus F_{\rm UBQC,out}\oplus
F_{\rm BPBO,static}.
\]
For computational-basis output, the $X$ components of this Pauli frame relabel
bit strings; the $Z$ components are phase-frame data unless coherent output or
non-$Z$ readout is requested.  Sec.~\ref{sec:results} shows the histogram
consequence of omitting a required term.

\begin{theorem}[Blindness at optimized leakage]\label{thm:blind}
Let an optimized pattern be produced by rewrites that encode no computation in
the topology: once the (permitted-leakage) dimensions $(n,m')$ are fixed, the
materialized graph is the \emph{standard} brickwork$(n,m')$, angles are sent in
the standard column order, any recycled-window schedule is a deterministic
function of $(n,m')$, and no rewrite metadata is sent to the server.  All
remaining computation dependence lies in the measurement angles, whose
adaptations are deterministic, client-computable, and $\Abfk$-valued. Under
UBQC blinding ($\delta=\phi'+\theta+r\pi$ with fresh uniform $\theta,r$ per
qubit), any two computations with the same admitted leakage $(n,m')$ induce the
same server-visible protocol-view distribution: materialized graph,
received-qubit ensemble, transmitted angle sequence, and outcome interaction
depend only on $(n,m')$. The only computation-dependent quantity intentionally
declared outside the angle one-time pad is the dimension $m'$ itself; size
hiding requires a bucket policy fixed independently of the hidden computation.
\end{theorem}

With branch frames closed, blindness reduces to a transcript question: after
the client computes each outcome-dependent adapted angle, applies fresh UBQC
pads, and transmits the result, does the server see exactly a standard UBQC
transcript at the declared dimensions?  Theorem~\ref{thm:blind}'s hypothesis is
BPBO's design discipline restated: materialization produces only the standard
brickwork at the leaked dimensions, and selected-region data stay below the
blinding layer.  The server learns the declared optimized dimensions $(n,m')$;
within a fixed leakage class, blindness is inherited from the base protocol
after these admission conditions establish transcript equivalence. Optimizer
logs, witness identifiers, noncanonical timing records, retry behavior, and
artifact handles are not server-visible protocol messages unless a
particular implementation sends them; they are outside the transcript modeled
by Theorem~\ref{thm:blind}.

\begin{lemma}[Round indistinguishability and verification compatibility]\label{lem:v1}
Interleave optimized computation rounds with test rounds (every qubit a dummy
$\ket{z}$, $z$ uniform, or a trap $\ket{+_\theta}$, $\theta$ uniform),
\emph{provided}: (C1)~test rounds are generated at the optimized dimensions
$(n,m')$; (C2)~trap/dummy positions are drawn, conditional on $(n,m')$,
independently of the computation and of optimizer-internal metadata such as
deposit ledgers, witness choices, and selected-region plans; (C3)~both round
types use the standard brickwork$(n,m')$, the same canonical column order, and
any recycled-window schedule as a deterministic function of $(n,m')$; and
(C4)~both round types are blinded by the same mechanism, so the base protocol's
preparation indistinguishability applies. Then the server's pre-abort
protocol transcript is distributed as in the base test-round protocol. Under
C1--C4, BPBO outputs are compatible with the hypotheses of the underlying
test-round verification theorem; soundness remains the cited base theorem.
\end{lemma}

Lemma~\ref{lem:v1} extends the same factoring to verification, and its
conditions are operational rather than technical: (C1) is met by generating
test rounds at the optimized dimensions, (C3) is the materializer's
mechanically checked invariant, and (C4) is the base protocol's own lemma.
(C2) is the one genuinely \emph{new} design constraint that optimization
introduces: trap placement must not be correlated with the computation or with
optimizer metadata---in particular, not with the deposit ledger, however
tempting it is to guard the ``load-bearing'' angles---since any such
correlation opens a leakage channel through test-round statistics. These are
compatibility conditions, not a new trap-soundness theorem. Under C1--C4 the
verified protocol can use the optimized geometry uniformly: every round, test
or computation alike, shrinks by the same factor (quantified in
Sec.~\ref{sec:results}).

\section{Implementation and Artifact-Gated Pipeline}\label{sec:pipeline}

Sections~\ref{sec:bpbo}--\ref{sec:security} define the BPBO contract,
certificates, and admission rule.  This section describes the submitted
artifact implementation of that contract.  Its semantic verifiers and
admission predicates are shared across the supported L1/L2/L3 candidate
certificates they inspect, while candidate generation is intentionally
registry-seeded.  Registration is only a source of candidates: a replacement
is executed only after it passes the semantic, frame, pattern, cost, and plan
gates below.  Unsupported or non-admitted regions fall back to the
unoptimized BFK09 materialization.

\noindent\emph{Evidence vocabulary.}
We label evidence as analytic proof, exact cyclotomic certificate,
numerical/tolerance reconstruction, sampled branch replay, runtime admission,
or regression equivalence.  Section~\ref{sec:results} uses these labels rather
than treating every artifact check as the same kind of evidence.

\subsection{Candidate Sources and Witness Registry}
Production circuits enter the implementation as \emph{basis streams} already
lowered to $\{H,\mathrm{CX},T,T^\dagger,\text{Paulis}\}$.  Since such streams
do not necessarily retain gate-level markers for multi-qubit cores, the
converter folds cores back semantically.  A window folds only when its
$8\times 8$ unitary equals $G\,X^{t}Z^{z}\,\CCZ\,G'$ exactly for per-wire
Hadamard gauges $G,G'$ and a Pauli correction.  The criterion is
decomposition-agnostic: any correct $\CCZ$/Toffoli expansion may fold, not
only a memorized template.  Each accepted fold carries a per-window equality
assertion, and the converted stream must recompose to the input at fidelity
$1.0$.

The witness registry is a candidate source and cache, not the definition of
BPBO.  It currently contains the clean $\CCZ$ three-cell witness
(\texttt{WIT-CCZ3}), the $H^{\otimes 3}\CCZ$ application witness
(\texttt{WIT-GROVER-BLOCK}), the endpoint-target CCX/Toffoli three-cell
application witness (\texttt{WIT-CCX-TARGET2}), the fixed middle-target
four-cell fallback (\texttt{WIT-CCX4}), and the composite four-application
Grover-3 pack (\texttt{PACK-GROVER3}).  The last entry is not a new primitive
witness; it is a registered composition with fixed output-frame decoding.
Regions without an executable registered candidate can still be
floor-certified and displayed as analysis results, but registry membership or
preview status alone never executes a replacement.

\subsection{Production Materializers}
The implementation has two materialization surfaces.  L1 and L2 run as a
cell-loop over the current operation stream: same-wire and adjacent two-wire
regions are certified, candidate replacements are constructed by the
registered families of Sec.~\ref{sec:layers}, and the local admission
predicate accepts only frame-safe witnesses.  L3 uses separate N3 candidates:
semantic CCZ/CCX-class regions are detected in the basis stream, before local
rewrites that would destroy the recoverable three-wire core, matched to
registered three-wire witnesses or composite packs, and selected only when the
resulting BFK09 pattern is smaller than the current materialization.

The materializer emits an execution plan rather than a mere pass log.  The
plan records ladder/admission status, stable handle, selected-for-execution
flag, backend, frame injections, decoder metadata, baseline/executed geometry,
materialized column spans, vertices, measured vertices, and the base-angle
table over $\Abfk$ under the deterministic schedule.  Preview regions explain
missed opportunities, but the accepted materialized pattern is the
authoritative execution object.

\subsection{Validation Gates}
Every submitted artifact payload is gated by checks tied to the contract
above: semantic equality or a registered witness certificate; declared
recomposition equality or tolerance; zero-branch semantics and branch-frame
closure; runtime-safe output frames; a standard-brickwork base-angle pattern
over $\Abfk$; positive materialized execution saving; and a plan whose
\texttt{execution\_plan.executed\_pattern} matches the runtime
\texttt{/phase/pattern}.  The decoder composes the per-shot branch byproduct,
the UBQC output one-time pad, and any static BPBO output frame.  Pattern-shape
checks enforce Theorem~\ref{thm:blind}'s hypotheses, while module batteries and
witness re-verification connect the implementation to App.~\ref{app:repro}.

This is the implementation boundary used in Sec.~\ref{sec:results}: submitted
lowered input streams receive the supported local analysis and certification,
but the submitted artifact optimizes execution only for the registry-seeded
candidate family.
Extending that family is a synthesis and materializer problem
(Sec.~\ref{sec:discussion}), not a change to the BPBO contract or to the
inherited UBQC security model.

\subsection{Cost Model and Selection}
We report four costs, because they answer different questions.  The executed
geometry is \emph{rows} times \emph{columns}; columns are the closest proxy
for measurement depth and communication rounds, while
\emph{vertices}$=$rows$\times$columns counts prepared qubits in the
full-pattern view.  \emph{Measured vertices} exclude the output column and
therefore count adaptive measurement steps.  \emph{Operation cells} count the
compiler-level regions selected before materialization and explain
\emph{why} a pattern shrank; they are not the executed geometry.  The
qubit-recycling platform adds a separate runtime-space quantity,
$n\times 2$ active qubits, determined by the two-column execution window.

Selection is therefore column-first.  A candidate may reduce operation cells
but still be rejected if its materialized pattern does not improve the
executed BFK09 geometry; a no-regression fallback then keeps the current
materialization.  Conversely, the authoritative data shown to the client and
in the certificate are the rows, columns, vertices, measured vertices, and
base angles of the accepted pattern itself.  The base-angle table is
client/certificate metadata used to generate blinded angles; optimizer tags,
candidate identifiers, and preview regions are not part of the server-visible
transcript.  This is why the execution plan is checked against the pattern
that actually runs: preview regions can explain an opportunity, but only the
accepted materialization is evidence of an optimized UBQC execution.

\section{Results}\label{sec:results}

\noindent\emph{Evidence map.}
The tables distinguish floor targets, exact witnesses, runtime-admitted
payloads, sampled simulator checks, and regression-equivalence evidence.  A
floor alone is not an executable-layout claim, and regression equivalence is
engineering regression evidence rather than an independent semantic proof.

\subsection{Layer-Representative Results}
The results are organized by the certified resynthesis layers of
Sec.~\ref{sec:layers}, not by a single flagship circuit.  Table~\ref{tab:layer-results}
freezes the representative results used in this manuscript.  The first row
isolates the L1 witness-admission claim; the next two exercise two-wire and
mixed one-/two-wire machinery; the following two isolate the L3
primitive/application distinction; the last row uses Grover-3 as an integrated
four-application case.

\onecolumngrid
\begin{center}
\refstepcounter{table}\label{tab:layer-results}
\begin{minipage}{0.96\textwidth}\fontsize{7}{8}\selectfont
\textbf{TABLE~\thetable.} Layer-representative optimization results.  When
geometry is reported, baseline and executed geometries are written as
rows$\times$columns; vertices are rows$\times$columns.  The authoritative
executed geometry is the materialized BFK09 pattern, while operation-cell
counts explain the selected region plan.  The reference column may be a raw
geometry, prior baseline, or floor target; the evidence column gives status
and handle.
\end{minipage}\par\vspace{1pt}
{\fontsize{7}{8}\selectfont
\renewcommand{\arraystretch}{0.94}
\setlength{\tabcolsep}{2.5pt}
\begin{ruledtabular}
\begin{tabular}{p{0.14\textwidth}p{0.18\textwidth}p{0.15\textwidth}p{0.17\textwidth}p{0.26\textwidth}}
Role & Example & Reference/baseline & Materialized/result & Evidence status and handle \\
\hline
L1 witness & one-wire SYNTH1Q & H-count floor $\lceil h/2\rceil$
& admitted one-brick BFK09 witness
& \texttt{HCOUNT-1W} exact reachability; SYNTH1Q branch-frame admission;
executable subset only \\
L2 sanity & Bell/CX & $2\times 13$ (26 vertices) & $2\times 5$ (10 vertices)
& \texttt{EQUIV-GATE} regression equivalence; compact two-wire execution \\
L1/L2 mixed & Grover-2 & $2\times 125$ (250 vertices) & $2\times 29$ (58 vertices)
& one-/two-wire certified resynthesis plus \texttt{EQUIV-GATE} materialization check \\
L3 primitive & clean $\CCZ$ & phase-gadget floor target $=3$ cells
& witness patch $3\times 25$ (75 vertices, 72 measured)
& \texttt{WIT-CCZ3}/\texttt{BRANCH-CLOSURE}; standalone CCZ is synthesis-available /
preview-safe unless selected by a runtime materializer \\
L3 application & endpoint CCX/Toffoli & endpoint floor target $=3$ cells
& executable payload $3\times 33$ (99 vertices, 96 measured)
& \texttt{WIT-CCX-TARGET2} exact witness plus branch replay; core witness
$3\times 25$, frame $Z\otimes Z\otimes I$ \\
Integrated & Grover-3 & $3\times 725$ (2175 vertices)
& $3\times 98$ (294 vertices, 291 measured), 12 cells
& \texttt{PACK-GROVER3}/\texttt{RUN-GROVER3} runtime-admitted registered
composite; \texttt{STAT-GROVER3} ideal $P(111)=0.9453125$ \\
\end{tabular}
\end{ruledtabular}}
\end{center}

\subsection{Benchmark Coverage and Fallback Controls}
The representative table above isolates the method's layers.  To answer the
separate question of breadth, the artifact package freezes a benchmark/control
matrix under handle \texttt{BENCH-CORPUS}.  This matrix is deliberately scoped:
it is a regression and evidence corpus for the implemented optimizer, not a
statistical optimality claim over all Clifford+$T$ inputs.  It combines the
ten-circuit \texttt{EQUIV-GATE} corpus with L3 controls that would otherwise be
easy to overstate: standalone clean $\CCZ$ is a preview-safe primitive unless a
runtime materializer selects it, endpoint-target CCX is runtime-admitted at
three macrocells, and fixed middle-target CCX remains a four-cell fallback.

\begin{center}
\refstepcounter{table}\label{tab:benchmark-coverage}
\begin{minipage}{0.96\textwidth}\fontsize{7}{8}\selectfont
\textbf{TABLE~\thetable.} Frozen benchmark and control coverage.
Baseline/executed geometries are rows$\times$columns when both are present.
The random rows are seeded Clifford+$T$ circuits from the \texttt{EQUIV-GATE}
corpus; their role is implementation coverage, not a distributional
performance claim.  The corpus has fallback and preview controls, but not a
broad no-op or admission-reject battery.
\end{minipage}\par\vspace{1pt}
{\fontsize{7}{8}\selectfont
\renewcommand{\arraystretch}{0.94}
\setlength{\tabcolsep}{2.5pt}
\begin{ruledtabular}
\begin{tabular}{p{0.17\textwidth}p{0.24\textwidth}p{0.17\textwidth}p{0.18\textwidth}p{0.18\textwidth}}
Group & Cases & Baseline $\to$ executed & Gate/control result & Scope \\
\hline
Positive controls
& Bell/CX; Grover-2; Grover-3
& $2{\times}13\to2{\times}5$;
  $2{\times}125\to2{\times}29$;
  $3{\times}725\to3{\times}98$
& \texttt{EQUIV-GATE} PASS; \texttt{PACK-GROVER3}/\texttt{RUN-GROVER3} PASS
& Layer demos plus registered integrated stress test \\
Seeded random 2q
& \texttt{rand2q\_s1}, \texttt{rand2q\_s2}
& $2{\times}93\to2{\times}45$;
  $2{\times}109\to2{\times}37$
& \texttt{EQUIV-GATE} PASS
& Two-wire corpus coverage \\
Seeded random 3q
& \texttt{rand3q\_s1}, \texttt{rand3q\_s2}
& $3{\times}133\to3{\times}69$;
  $3{\times}117\to3{\times}45$
& \texttt{EQUIV-GATE} PASS
& Three-wire non-specialized inputs \\
Seeded random 4q
& \texttt{rand4q\_s1}, \texttt{rand4q\_s2}
& $4{\times}109\to4{\times}61$;
  $4{\times}117\to4{\times}53$
& \texttt{EQUIV-GATE} PASS
& Wider compiler path; no L3 optimality claim \\
Endpoint CCX
& target row $2$
& floor target $3$ cells $\to3{\times}33$
& \texttt{WIT-CCX-TARGET2} PASS
& Runtime-admitted L3 application \\
Standalone $\CCZ$
& clean $\CCZ$ witness
& floor target $3$ cells $\to3{\times}25$ patch
& \texttt{WIT-CCZ3}/\texttt{BRANCH-CLOSURE} PASS
& Synthesis-available / preview-safe primitive \\
Middle-target CCX
& target row $1$
& floor target $3$ cells; known executable $4$ cells
& \texttt{WIT-CCX4} truth-table PASS
& Row-asymmetry fallback, not a silent replacement \\
$n{=}4$ platform smoke
& three random branch-pattern trials
& $4{\times}24$ window; peak active window $=8$
& \texttt{SMOKE-N4} PASS
& Platform-only smoke test; no optimization claim \\
\end{tabular}
\end{ruledtabular}}
\end{center}
\twocolumngrid

The current package includes fallback and preview-safe controls, but not a
broad admission-reject/no-op corpus.  This limits empirical claims about
rejection behavior on unstructured inputs.  The correctness claim instead
rests on the admission verifier and the materialized-pattern checks of
Sec.~\ref{sec:pipeline}; broader negative testing is a useful artifact
extension and is listed as a limitation in Sec.~\ref{sec:discussion}.

\subsection{Integrated Application: Grover-3}
We use three-qubit Grover search~\cite{grover96} (two iterations, marked state
$\ket{111}$) as a registered integrated stress test, not as a broad benchmark.
The production runtime receives only the lowered basis stream
$\{H,\mathrm{CX},T,T^{\dagger},X\}$ ($95$ gates).  The converter folds four
$\CCZ$-class subsequences semantically ($95\to 39$ abstract gates), records
four floor-3 cores and a fidelity-$1.0$ recomposition check, and gates
execution on \texttt{PACK-GROVER3}.  The runtime-admitted materialization
\texttt{RUN-GROVER3} is
\[
3 \times 98 = 294 \text{ vertices},
\]
versus $3\times 301 = 903$ for the prior certificate-era optimizer and
$3\times 725 = 2175$ for the raw lowering (Fig.~\ref{fig:columns}): a
$3.07\times$ reduction over the optimized baseline and $7.4\times$ overall.
In the dynamic-circuit schedule of Table~\ref{tab:resources}, this is
$291$ mid-circuit measurements/reset opportunities and $97$ adaptive column
cycles on six active qubits, versus $2172$ measurements and $724$ cycles for
the raw streamed reference.
The \texttt{RUN-GROVER3} server probe (fold count, certification, runtime
admission, and materialization) records $181.819$\,s, reported as $182$\,s;
the separate full-stack harness records its own timing.  The measured column
count lies in the predicted $97$--$105$ band: the twelve-cell core is exact,
the four blocks concatenate at $24\equiv 0 \pmod 8$ columns each, and Pauli
layers are represented in frame/decoder metadata rather than extra columns.

\begin{figure}[t]
\centering
\begin{tikzpicture}
\begin{axis}[width=0.92\columnwidth, height=4.6cm, ybar,
  bar width=9mm, ymin=0, ymax=2500,
  ylabel={brickwork vertices}, ylabel near ticks,
  symbolic x coords={raw, cert, bpbo},
  xtick=data, enlarge x limits=0.28,
  xticklabels={{raw lowering\\$3\times725$},
               {two-wire certified\\$3\times301$},
               {this work\\$3\times98$}},
  xticklabel style={align=center, font=\scriptsize},
  yticklabel style={font=\scriptsize},
  ylabel style={font=\scriptsize},
  nodes near coords, every node near coord/.append style={font=\scriptsize},
  axis lines*=left, clip=false]
\addplot[fill=black!55, draw=black] coordinates
  {(raw,2175) (cert,903) (bpbo,294)};
\node[font=\scriptsize, anchor=west] at (axis cs:cert,2175)
  {$7.4\times$ overall, $3.07\times$ vs.\ baseline};
\end{axis}
\end{tikzpicture}
\caption{Grover-3 materialization cost at the three stages of optimization:
raw Clifford+T lowering onto the brickwork (geometry $3\times 725$, i.e.,
$725$ columns on three rows), the prior two-wire certificate optimizer
($3\times 301$), and the full pipeline of this work ($3\times 98$,
\texttt{RUN-GROVER3})---a $7.4\times$ overall reduction. Each accepted pattern
remains in the standard brickwork family at its declared optimized dimension;
the leaked dimension $(n,m')$ is permitted leakage under
Theorem~\ref{thm:blind}.}
\label{fig:columns}
\end{figure}
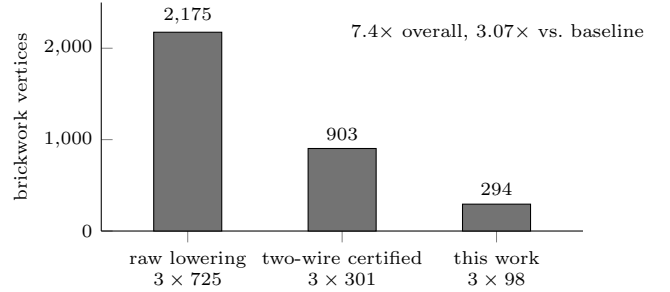

\subsection{Measured Output Statistics}
We also check outcome-level behavior by measurement.  The frozen
\texttt{STAT-GROVER3} run executes the optimized pattern under the
computation-round recycled-window simulator with Born sampling, adaptive
corrections, and the composed output decoder; over $4000$ shots it yields
\[
P(\ket{111}) = 0.9445 \quad\text{vs.\ ideal } 0.9453
\]
($0.2\sigma$ for the marked bin at binomial $\sigma=0.0036$), with
total-variation distance $0.0028$ to the ideal eight-outcome distribution
(Fig.~\ref{fig:histogram}). Every shot ran with six active logical qubits
(Sec.~\ref{sec:platform}). The production runner, transfer-matrix machinery,
and an independently coded window sampler agree on the relevant reference
checks, so the histogram is a consistency check for the admitted pattern and
decoder, not an independent proof of blindness or verifiability.

\begin{figure}[t]
\centering
\begin{tikzpicture}
\begin{axis}[width=0.95\columnwidth, height=5.0cm, ybar,
  bar width=2.9mm, ymode=log, log origin=infty, ymin=0.003, ymax=2.2,
  ylabel={probability (log)}, ylabel near ticks,
  symbolic x coords={000,001,010,011,100,101,110,111},
  xtick=data, enlarge x limits=0.07,
  xticklabel style={font=\scriptsize}, yticklabel style={font=\scriptsize},
  ylabel style={font=\scriptsize},
  legend style={font=\scriptsize, at={(0.04,0.96)}, anchor=north west,
                draw=none, fill=none},
  legend cell align=left, axis lines*=left]
\addplot[fill=black!60, draw=black,
  error bars/.cd, y dir=both, y explicit,
  error bar style={black}] coordinates {
  (000,0.00725) +- (0,0.00134)
  (001,0.0095)  +- (0,0.00153)
  (010,0.0075)  +- (0,0.00136)
  (011,0.0085)  +- (0,0.00145)
  (100,0.00825) +- (0,0.00143)
  (101,0.0075)  +- (0,0.00136)
  (110,0.007)   +- (0,0.00132)
  (111,0.9445)  +- (0,0.00362)};
\addplot[fill=white, draw=black] coordinates {
  (000,0.0078125) (001,0.0078125) (010,0.0078125) (011,0.0078125)
  (100,0.0078125) (101,0.0078125) (110,0.0078125) (111,0.9453125)};
\legend{measured (4000 shots), ideal}
\end{axis}
\end{tikzpicture}
\caption{Decoded output distribution of the optimized Grover-3 pattern under
the computation-round recycled-window simulator (\texttt{STAT-GROVER3},
$4000$ Born-sampled shots, six active logical qubits), against the ideal
two-iteration Grover distribution. The
marked outcome $\ket{111}$ is recovered at $0.9445$ versus the ideal $0.9453$
($0.2\sigma$); total-variation distance $0.0028$. Error bars are
$\pm 1\sigma$ binomial at $4000$ shots; the vertical scale is logarithmic so
the seven unmarked outcomes remain resolvable. Decoding applies the
per-shot branch byproduct mask, the UBQC output $X$ pad, and the static BPBO
output frame; omitting a required term relabels the histogram.}
\label{fig:histogram}
\end{figure}
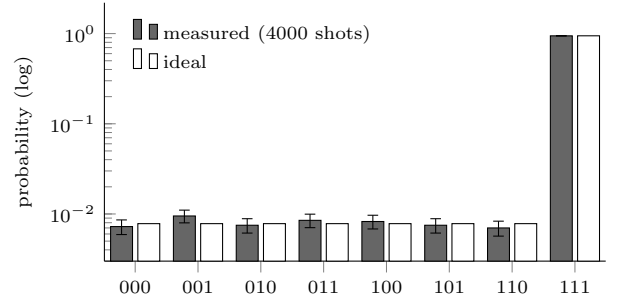

\subsection{Inherited Test-Round Cost Profile}
Under Lemma~\ref{lem:v1}, inherited test-round verification can be run on the
optimized geometry at \emph{reduced} absolute cost: every round---test or
computation---is a $3\times 98$ pattern, so the $\sim 3\times$ per-round saving
multiplies through the protocol's round overhead. On this geometry, the
\texttt{TRAP-REF} harness uses uniform traps on one bipartite class at density
$1/2$, giving $T\approx 73$ traps per round in expectation over the measured
columns.  Its analytic curve uses
\(1-(1-\sin^2(\varepsilon/2))^T\) and matches the frozen simulations across
$\varepsilon\in[0.05,0.2]$; a single-column $Z$-tamper model is detected with
probability $0.63$.  Since execution is column-streamed, rejected rounds can
abort when the first violation appears, saving $51$--$56\%$ of the remaining
columns conditional on rejection.  Soundness amplification remains the base
protocol's theorem; these measurements show the cost profile on the optimized
geometry.

\subsection{Endpoint Toffoli and the Layout Gap}
The L3 application row of Table~\ref{tab:layer-results} isolates the
Toffoli/CCX point that is easiest to misstate.  The floor algorithm certifies
$\mathrm{floor}(\mathrm{CCX})=3$, but floor certification and executable
layout are separate questions.  For the endpoint-target placement used by
the production Toffoli path (controls on rows $0,1$, target on row $2$), the
registered \texttt{WIT-CCX-TARGET2} witness closes at the floor: three
macrocells, connected core $3\times 25$ ($3\cdot 8+1$ columns), output frame
$Z\otimes Z\otimes I$, exact cyclotomic zero-branch certification, and
endpoint branch replay.  The submitted executable payload then adds the
standard eight-column runtime wrapper/boundary around that connected core,
giving the displayed $3\times 33$ pattern with $99$ vertices and $96$ measured
vertices; this payload dimension is not claimed as a separate minimal-layout
theorem.

The fixed middle-target placement remains deliberately separate.  The same
floor value is three, but the known middle-target witness is the registered
\texttt{WIT-CCX4} four-cell fallback with frame $X\otimes Y\otimes Y$.  This is not a
contradiction and not a silent replacement policy: the BFK09 three-row
geometry is not row-symmetric, so the endpoint-target application can absorb
the boundary Hadamards at the floor while the fixed middle target currently
uses the conservative four-cell path.  In the paper's main result table,
therefore, endpoint-target CCX/Toffoli is the L3 application result; the
middle-target witness is a scope and robustness item for the witness appendix.
The relation to adjacent systems is summarized in Sec.~\ref{sec:related}.

\subsection{Regression Equivalence of the Unified Loop}
The unified orchestration of Sec.~\ref{sec:pipeline} replaced the
development-era fixed pass chain, and the replacement is gated by a
reproducible regression-equivalence harness rather than by code inspection
(handle \texttt{EQUIV-GATE}, App.~\ref{app:repro}). On the frozen ten-circuit
corpus---Bell, Grover-2, Grover-3, the legacy fixed-middle Toffoli row, and
seeded random Clifford+$T$ circuits on two to four wires---the gate compares
declared fields against the pre-consolidation runtime: (i) materialized
pattern geometry, where the loop must do no worse and in fact reproduces every
column count exactly, including Grover-3 at $98$; (ii) the frozen
expected-output fields under the same decoder convention; and (iii)
per-family rewrite counters, pinning not just the result but the applied
rewrite families.  All three invariants hold on all ten circuits.  This is
regression evidence, not an independent semantic proof; the exact witnesses,
branch replay, runtime admission, and sampled statistics above provide the
claim-specific evidence.
The current endpoint-target \texttt{WIT-CCX-TARGET2} Toffoli path is
validated by the separate L3 witness/materializer battery summarized in
Table~\ref{tab:layer-results}; it is not conflated with the legacy
middle-layout row.

\section{Related Work}\label{sec:related}

\emph{Blind and verifiable quantum computation.} UBQC originates with the
brickwork protocol of Broadbent, Fitzsimons, and Kashefi~\cite{bfk09}, built
on the one-way model~\cite{rb01} and the measurement
calculus~\cite{dkp07}---our byproduct trackers are instances of the latter's
signal-shifting algebra. Verifiability was added through trap-based
protocols~\cite{fk17} and brought to practical overheads by test-round
schemes~\cite{leichtle21} (see~\cite{gkk19} for a survey); experimentally,
blind delegation was first demonstrated photonically~\cite{barz12}, and
verifiable blind delegation has been demonstrated on a networked trapped-ion
server with a photonic client~\cite{drmota24} and in a multi-client Qline
configuration~\cite{polacchi25}. In the BFK/prepare-and-send brickwork line,
the public graph is computation-independent within the declared leakage
surface. Overheads have been optimized at the protocol
level---the quantum communication of blind computation has been upper- and
lower-bounded, placing the brickwork protocol within a constant factor of
optimal for single-qubit clients~\cite{mantri13}---but that line bounds the
\emph{protocol}, uniformly; it does not provide a layer that
certifies how small the fixed-brickwork pattern of a \emph{given} computation
may be. (A
complementary line removes the quantum client
entirely under computational assumptions~\cite{mahadev18}; the present work
stays information-theoretic, in the prepare-and-send model where the
brickwork constraint---and hence the size-optimization problem---arises.)
BPBO contributes such a certifying optimization layer \emph{inside} the
constraint these protocols impose, with admission and compatibility arguments
(Theorem~\ref{thm:blind} and Lemma~\ref{lem:v1}) showing that admitted
rewrites compose with inherited blindness and verification guarantees rather
than re-deriving either.

\emph{MBQC pattern optimization.} Standardization~\cite{dkp07},
Pauli-flow preprocessing and circuit extraction~\cite{simmons21},
flow-preserving rewrite systems for Pauli-measurement patterns~\cite{mcelvanney22},
and graph-level optimizers such as Graphix~\cite{bkmp07,graphix} optimize or
reorganize MBQC representations by allowing changes to the command structure,
open graph, or a closely related graph-like representation; ZX-based circuit
simplification uses local complementation and pivoting in graph-like
diagrams~\cite{duncan20}. These methods are powerful when the graph or
graph-like representation is an optimization variable; under the fixed BFK09
leakage model, computation-dependent graph changes are not directly admissible
unless followed by padding or a security argument showing that the rewritten
graph leaks no computation-dependent structure. BPBO is complementary: it keeps
the output inside the standard brickwork family at declared optimized
dimensions, changes hidden measurement-angle data and declared permitted
length leakage, and certifies floors or executable witnesses over the
reachable BFK family $R_{\mathrm{BFK}}$ rather than over all graphs.

\emph{Brickwork resource reduction.} A line specific to UBQC lowers the
cost of the brickwork realization itself: Chien, Van Meter, and
Kuo~\cite{chien13} count the fault-tolerant brickwork layers of a circuit by
tiling each gate of a Clifford+T decomposition into bricks, while a more
recent thread reduces the qubit count by \emph{altering} the resource
state---substituting smaller cluster fragments for selected
bricks~\cite{yang22}, or replacing the fixed brickwork with non-fixed
variants that cut the ancillae per gate~\cite{ma24}. These approaches optimize
different resource axes: decompose-and-tile analyses the cost of standard
brickwork realizations, while resource-state variants change the graph and
therefore require their own leakage argument or padding discipline under a
BFK-style interface. BPBO instead preserves the standard brickwork family at
declared optimized dimensions and shortens admitted patterns by
arity-stratified certified resynthesis; the logical qubit-recycled runner then
executes the resulting standard brickwork pattern at constant logical active
width.

\emph{Circuit-model optimization and synthesis.} In the circuit model, many
resource measures are optimized under their own cost models, including exact
and approximate Clifford+T synthesis~\cite{kmm13,rs16}, T-count and T-depth
optimization via phase polynomials over $\mathbb{F}_2$-parities~\cite{amm14},
and entangling-gate lower bounds for unitary synthesis~\cite{sbm06}. Our
demand-side machinery adapts part of this reasoning to the blind setting: the
parity ledger of Sec.~\ref{sec:certificates} is a phase-polynomial and
phase-gadget representation~\cite{amm14,cowtan20}, and the L1/L2/L3 floors are analogues of
single-qubit, entangling-count, and three-wire cell-complexity bounds. The
L3 certificates borrow phase-gadget language from this synthesis tradition,
but the admitted objects are fixed-brickwork MBQC witnesses rather than
circuit-level rewrites. The
differences arise from the fixed-brickwork UBQC interface: supply is a fixed
brickwork cell menu rather than freely placed gates, residues chain through
Clifford frames across cells, and the output is a certificate that gates
execution rather than a circuit-level rewrite. The two
toolchains compose: circuit-level optimization applies upstream, before
lowering, and BPBO compresses certifiable local slack left by the lowering on
admitted regions.

\emph{Qubit reuse.} Measure-early-reset-and-reuse compilation is established
for circuit-model hardware, where finding the reuse schedule is itself an
optimization problem~\cite{decross23}. The runner of
Sec.~\ref{sec:platform} is the brickwork counterpart: column ordering makes
the two-column live window canonical for the logical brickwork dependency
structure (Table~\ref{tab:resources}), and the same lowering yields logical
dynamic-circuit schedules with $n\times$window active qubits. Earlier
circuit-based UBQC/MBQC simulation work studied small blind
computations on gate-model platforms, including a two-qubit Grover
instance~\cite{leechung25}. Here that execution viewpoint is used as context;
the present paper centers on BPBO-certified local reduction and
qubit-recycled execution of optimized standard brickwork patterns. What
distinguishes the runner from both lines is the coupling: constant-width
logical execution preserves the UBQC transcript interfaces---blinding,
byproduct decoding, and test-round compatibility under Sec.~\ref{sec:security}'s
conditions---which Sec.~\ref{sec:results} exercises through artifact-validated
simulation, statistics, and regression gates.

\emph{Design-space summary.} Table~\ref{tab:comparison} summarizes the
combination addressed here.  It is not a ranking of systems with different
goals: BPBO is compared as a certifying optimizer for the standard brickwork
family at declared optimized dimensions, while the qubit-recycled runner is
compared as a logical execution stack for the optimized standard brickwork
patterns.  The non-affirmative entries identify scope differences, not
deficits.

\onecolumngrid
\begin{center}
\refstepcounter{table}\label{tab:comparison}
\begin{minipage}{0.96\textwidth}\footnotesize
\textbf{TABLE~\thetable.} Representative systems from the two nearest lines
of work against the six properties addressed by BPBO and the qubit-recycled
execution stack.  This is a design-space summary, not a general-purpose
ranking: the compared systems pursue orthogonal goals, and each
non-affirmative entry names the reason, scoped to the cited version (Graphix
as described in~\cite{graphix}; qubit-reuse compilation as
in~\cite{decross23}) and footnoted.  The last row summarizes the combination
addressed in this paper.
\end{minipage}\par\smallskip
\footnotesize
\resizebox{0.98\textwidth}{!}{%
\begin{tabular}{lcccccc}
\toprule
 & graph-preserving & size-floor & blindness & verification & constant-width
 & artifact-validated \\
 & optimization & certificates & preservation & compatibility
 & UBQC execution & optimized execution \\
\midrule
Graphix~\cite{graphix} & graph-changing\textsuperscript{a}
 & not targeted\textsuperscript{b}
 & out of scope\textsuperscript{c} & out of scope\textsuperscript{c}
 & no UBQC transcript stack\textsuperscript{d} & ---\textsuperscript{d} \\
Qubit-reuse compilation~\cite{decross23} & circuit-model\textsuperscript{e}
 & not targeted\textsuperscript{b} & out of scope\textsuperscript{c}
 & out of scope\textsuperscript{c}
 & no UBQC transcript stack\textsuperscript{d} & ---\textsuperscript{d} \\
This work: BPBO + recycled execution & Yes (Sec.~\ref{sec:bpbo}) & Yes (Sec.~\ref{sec:certificates})
 & Yes (Thm.~\ref{thm:blind}) & Yes (Lem.~\ref{lem:v1})
 & Yes (Sec.~\ref{sec:platform}) & Yes (Sec.~\ref{sec:results}) \\
\bottomrule
\end{tabular}}
\par\smallskip
\begin{minipage}{0.98\textwidth}\raggedright\scriptsize
\textsuperscript{a}\,Optimizes by rewriting the open graph into
local-Clifford decorated form---a graph change that is not directly
admissible under fixed-topology blindness without padding or a separate
leakage argument (Sec.~\ref{sec:related}).
\textsuperscript{b}\,Emits optimized artifacts, not lower-bound
certificates on pattern size.
\textsuperscript{c}\,Blind or verified delegation is outside the system's
stated scope; no security claim is made or required there.
\textsuperscript{d}\,Neither system targets the UBQC transcript/admission stack
(blinding, byproduct decoding, and test-round compatibility), which these axes
require; qubit-reuse compilation targets circuit-model hardware.
\textsuperscript{e}\,Operates in the circuit model; the graph-preservation
axis does not apply.
\end{minipage}
\end{center}
\clearpage
\twocolumngrid

\section{Scope and Limitations}\label{sec:discussion}

\emph{Scoped optimality claims.} Theorem~\ref{thm:nogo} excludes two-cell
CCZ realizations within the CNOT+T phase-gadget family; the three-cell
witness lies in the same family, so the clean-window cell complexity is
exactly three within that family. The same witness gives an unconditional
three-cell upper bound in the full brickwork model under the stated macro-cell
convention. Outside the family we have adversarial evidence only---approximately
$1.8\times 10^{5}$ randomized searches over arbitrary-angle two-cell nets,
with frame-aligned fidelity reaching at most $0.854$---and a family-free
two-cell no-go remains open. Thus every formal ``exactly three'' statement in
Sec.~\ref{sec:ccz} carries the phase-gadget scope explicitly.

\emph{General verification, registry-seeded synthesis.} The verifier and
admission predicates are schema-general for submitted local candidates, but
candidate generation in the submitted artifact is registry-seeded. The L3
registry covers the clean $\CCZ$ witness, the $H^{\otimes 3}\CCZ$ application witness, the
endpoint-target $\mathrm{CCX}$ application witness, the fixed middle-target
fallback, and the registered Grover-3 composition
(Sec.~\ref{sec:results} and App.~\ref{app:witness}). The frame-chained
synthesizer is a heuristic with a sound verifier, not a deterministic or
complete witness compiler, so a local region without an admitted
witness/materialization retains the unoptimized BFK09 lowering. Closing this
gap---a synthesis-and-materialization procedure with completeness guarantees
over a specified $R_{\mathrm{BFK}}(3,k)$ target class, frame convention, and
admission contract---is a theory target, not a current claim.

\emph{Scaling beyond three wires.} The formal certificate language generalizes
as stated: the supplier side is governed by the binary symplectic
representation of Clifford/stabilizer operations~\cite{aaronson04,gottesman98},
and the demand side by a $2^n{-}1$-coordinate parity ledger, so a floor is
definable once the target class, supply filtration, and frame quotient are
specified. What grows is not only search cost but executable coverage:
materializers, frame metadata, branch checks, and benchmarks must also scale.
Complete all-target enumeration is used through $n=2$; the $n=3$ claims here
are targeted floor batteries, schedule enumerations, and witness verifications
for the registered CCZ/CCX-class families, not a catalog of every three-wire
unitary. Wider targets will need structure-exploiting search in place of
enumeration; the floor algorithm's cap-and-certify design supports scoped
certificates, but no complete beyond-three-wire compiler or materializer
coverage is delivered here.

\emph{Benchmark and rejection coverage.} The benchmark suite is a fixed
regression and evidence corpus for the implemented optimizer, not a
distributional performance study over Clifford+$T$ inputs. The artifact
includes fallback and preview-safe controls, but not a broad no-op or
admission-reject corpus covering malformed candidates, invalid frames,
unsupported windows, invalid alphabets, or non-improving materializations.
Consequently, the empirical claims are limited to the representative L1/L2/L3
cases, seeded regression circuits, and explicit L3 fallback checks of
Sec.~\ref{sec:results}; broader negative testing is future artifact work.

\emph{Size leakage and padding.} Pattern length is treated throughout as
permitted leakage, matching the BFK09 model in which the server sees the graph
dimensions. A client who must also hide optimized-size information can pad or
bucket the optimized pattern to public standard-brickwork dimensions, trading
away some or all of the resource the optimization recovered. Theorem~\ref{thm:blind}
applies at whichever declared dimensions are chosen, provided the padded or
bucketed pattern still satisfies the admission and transcript conditions; the
choice belongs to the application policy, not to the BPBO methodology.

\emph{Verification constraints and evidence level.} Lemma~\ref{lem:v1}'s
condition C2 obliges trap and dummy positions to be drawn independently of
the computation---conditioned only on the declared dimensions $(n,m')$ and,
in particular, uncorrelated with deposit ledgers or other high-impact angle
locations. A generator tuned to guard those locations falls outside the
inherited-verification argument. Our verification evidence is simulation-level
mechanism and cost evidence: detection and abort statistics
(Sec.~\ref{sec:results}) match the $\sin^2(\varepsilon/2)$ prediction for the
stated perturbation model on the optimized recycled-window geometry, and no
new soundness theorem is claimed beyond the inherited guarantees under
Lemma~\ref{lem:v1}'s conditions~\cite{fk17,leichtle21}. Verifiable blind
delegation has been demonstrated on hardware at small scale~\cite{drmota24,polacchi25};
what has not been demonstrated is BPBO-optimized standard-brickwork execution
with the UBQC/recycled-execution stack, and the dynamic-circuit
lowering~\cite{decross23} makes that a concrete experimental follow-up.

\section{Conclusion}\label{sec:conclusion}

In prepare-and-send BFK09 brickwork, the public graph family and declared
dimensions constrain optimization, but they do not eliminate it. BPBO shows
that BFK09-compatible patterns can be shortened to declared optimized
standard-brickwork dimensions by certified local resynthesis: regions are
analyzed by arity, floors or witnesses are checked, and only
protocol-admissible materializations are executed. The qubit-recycling runner
is a separate logical execution stack that makes the resulting long blind
patterns executable and testable at constant logical active width. The
submitted artifact implementation is registry-seeded rather than complete, but
it exercises a common certificate-and-admission discipline across
representative L1 one-wire, L2 two-wire, and L3 three-wire cases, with
Grover-3 as the integrated application. The next steps are a complete witness
compiler for a specified target class, structure-exploiting search and
executable coverage beyond three wires, broader no-op/admission-reject
benchmarks, and hardware execution of BPBO-optimized standard-brickwork
patterns.

\begin{acknowledgments}
This work was supported by Electronics and Telecommunications Research
Institute (ETRI) grant funded by the Korean government [26ZS1320, Research on
Quantum-Based New Cryptographic System for Ensuring Perfect Data Privacy].
\end{acknowledgments}

\section*{Data Availability}
The artifact package supporting the BPBO and qubit-recycled execution claims is
available as ancillary material accompanying this arXiv submission and is
described in App.~\ref{app:repro}.  The public record is defined by the
ancillary directory \texttt{anc/artifact\_package/}, \texttt{manifest.json},
and \texttt{SHA256SUMS.txt}.  The package
contains frozen result data, registered witness tables, verification scripts,
environment pins, and the code-only runtime used for the quick and full
reproducibility checks.  Usage and public-archive license status are recorded
in the artifact's \texttt{USAGE\_AND\_LICENSE.txt}; a DOI-bearing archive can
be cited in later versions when available.  The principal witness angle tables are printed
in full in App.~\ref{app:witness}, so the central $\CCZ$, Grover-block, and
CCX witness/layout claims can be checked independently of the packaged runtime;
benchmark, statistical, resource-model, and full-stack equivalence claims are
reproduced from the artifact package.

\appendix
\section{Proofs}\label{app:proofs}

\noindent\emph{Conventions.}
Throughout Appendix~\ref{app:proofs}, phases in parity ledgers are measured
in $\pi/4$ units modulo $8$, equality of implemented unitaries is modulo a
global phase unless stated otherwise, and
$R_z(t)=\mathrm{diag}(1,e^{it})$. We use column-vector multiplication,
$\ket{+_t}=(\ket{0}+e^{it}\ket{1})/\sqrt2$, and
$P(a,b)=\prod_r X_r^{a_r}Z_r^{b_r}$ in the printed wire order. A branch
$+\pi s$ is projector semantics for the observed outcome, not an extra
server instruction.

\subsection{Parity-Ledger Necessity}
For a $\pi/4$-valued diagonal $3$-qubit phase function in the CNOT+$T$ /
phase-gadget scope, write the phase over parity forms:
$\varphi(x)=\tfrac{\pi}{4}\sum_{L} c_L\,\chi_L(x)$ with $c_L\in\mathbb{Z}_8$,
$L$ ranging over the seven nonempty subsets of $\{x_0,x_1,x_2\}$ and
$\chi_L(x)=\bigoplus_{i\in L}x_i$.

A gadget realization specifies its deposits explicitly, so it \emph{comes
with} a representation $\varphi=\sum_i a_i\,\chi_{L_i}$; no uniqueness of the
coefficients is claimed (indeed none holds mod $8$: the parity-evaluation
matrix has determinant $\pm 32$, and e.g.\
$4(\chi_A+\chi_B+\chi_{A\oplus B})\equiv 0 \bmod 8$). What is invariant is the
\emph{odd support}:

\begin{lemma}[Necessity: the odd-support syndrome]\label{lem:ledger}
For a phase-gadget realization with deposit multiset $\{(L_i,\,a_i\pi/4)\}$,
define the deposit syndrome $o\in\mathbb{F}_2^{7}$ by
$o_L=\sum_{i:L_i=L} a_i \bmod 2$. Then:
(i)~$o$ is an invariant of the realized diagonal \emph{up to global phase}:
any two integer representations of phase functions equal mod $8$ up to an
additive constant differ, beyond the constant coordinate, by a kernel vector
of the parity-evaluation matrix, and every such kernel vector is even---so the
\emph{nonconstant} odd-support syndrome is well defined on diagonals modulo
global phase.
(ii)~$\CCZ$ has syndrome $o=(1,\dots,1)$: the expansion of $\pi x_0x_1x_2$
over parities has coefficients $(1,1,1,7,7,7,1)$, all odd.
(iii)~The syndrome is unchanged by single-qubit Pauli output frames on
diagonal targets. Consequently every gadget realization of $\CCZ$ modulo a
single-qubit Pauli frame deposits odd phase on \emph{all seven} parities---in
particular on the three $x_0$-containing forms.
\end{lemma}

\begin{proof}
(i) The $8\times 8$ parity-evaluation matrix $M$ (columns: the constant and
the seven $\chi_L$; rows: the points of $\{0,1\}^3$) has Smith normal form
with invariant factors $(1,1,1,1,2,2,2,4)$; accordingly its kernel mod $8$ has
order $32$, isomorphic to $\mathbb{Z}_4\times\mathbb{Z}_2^{3}$, and consists
of even vectors only. Two of its relations are transparent:
$2\sum_{L}\chi_L\equiv 0$ (every nonzero point lies in exactly four of the
seven parities) and, for each variable $x_i$,
$4\sum_{L\ni x_i}\chi_L\equiv 0$ (every point lies in zero, two, or four of
the four parities containing $x_i$); the full kernel is tabulated in the
artifact (App.~\ref{app:repro}) and is reproduced by exhaustive enumeration.
A global phase shifts only the constant coordinate. Since every invariant
factor divides $4$, every kernel vector lies in $2\mathbb{Z}_8^{\,8}$; hence
representations of the same diagonal (mod global phase) agree mod $2$ on the
nonconstant coordinates, and $o$ is well defined.
(ii) Direct expansion (machine-verified):
$\pi x_0x_1x_2 = \tfrac{\pi}{4}\big[\textstyle\sum_i x_i
-\sum_{i<j}(x_i{\oplus}x_j) + (x_0{\oplus}x_1{\oplus}x_2)\big]$, i.e.\
$c=(1,1,1,7,7,7,1)$.
(iii) Write $P=X^aZ^b$. An $X$ output frame is a readout translation
$x\mapsto x\oplus a$, not a phase: writing the realization in monomial form
$X^{a}\,D$ with $D$ its diagonal factor, the frame equation
$X^aD=X^aZ^b\,\CCZ$ reduces to $D=Z^b\,\CCZ$. Translations send each parity to
itself up to a sign and a constant
($\chi_L(x\oplus a)=\chi_L(x)\oplus\chi_L(a)$; in the integer lift,
$u\oplus1=1-u=-u+1$ and $u\oplus0=u$), so they preserve the nonconstant
syndrome; admitting $X$ frames therefore does not relax the odd-support
requirement. For the diagonal equation, $Z^b$ multiplies the phase
by $\pi\sum_{i\in b}x_i$, an \emph{even} representation shift ($4$ per
affected single-variable parity). (Machine closure over all $64$ Pauli
frames.)
\end{proof}

\subsection{The Supply--Demand Floor (Theorem~\ref{thm:floor})}
\begin{proof}
Let $U$ admit a $k$-cell realization in the class $F$, i.e.\
$U\in\mathcal{R}_F(n,k)$ up to the model's gauge and Pauli-frame dressing.
Orbit-invariance of $D$ makes $D(U)$ well defined on the dressed target,
and the over-approximation property then places $D(U)\in S(k)$. By
definition of $\mathrm{floor}_D(U)$, $D(U)\notin S(j)$ for all
$j<\mathrm{floor}_D(U)$.  Hence no realization with fewer cells exists in
$F$: $k^F_{\min}(U)\ge\mathrm{floor}_D(U)$.  If the minimum is empty, then no
realizable $U$ in the stated class can have that demand, because any
realization would place it in some $S(k)$. Soundness of each
instantiation below therefore reduces to checking exactly two properties
of its $(D,S)$ pair---orbit-invariance of the demand and
over-approximation of the supply---which is how the proofs are organized.
\end{proof}

\subsection{Wire-Count Instantiations (Corollaries~\ref{thm:hcount},
\ref{thm:l2})}
\begin{proof}[Proof sketch (Corollary~\ref{thm:hcount})]
Both directions are general in $k$. \emph{Supply} (upper bound): one cell's
wire chain is $R_z H R_z H R_z$---exactly two Hadamards with free
$\Abfk$-phases---so $k$ concatenated cells realize, by construction, exactly
the alternating words with at most $2k$ Hadamards (adjacent phases merge).
\emph{Demand} (achievability): any $U$ with $h(U)\le 2k$ has a minimal-$H$
normal form; if its Hadamard count falls short of the supplied $2k$ by an odd
amount, the identity $H=S\,H\,S\,H\,S$ pads the word without changing the
realized unitary or exceeding the supply, and leftover slots absorb as
identity cells. Hence $\Rbfk{1}{k}=\{h\le 2k\}$ for every $k$. The exhaustive
enumeration in the cyclotomic ring $\mathbb{Z}[\zeta_{16}]$ (exact arithmetic,
no floating point) through $k=3$ is an independent confirmation of the
equality at small $k$, not the source of generality. Artifacts in
App.~\ref{app:repro}.
\end{proof}
\begin{proof}[Proof sketch (Corollary~\ref{thm:l2})]
A two-wire cell carries two rungs, hence entangling capacity two: $k$ cells
realize CNOT-cost at most $2k$, giving the floor $\lceil c/2\rceil$ with $c$
the Makhlin-invariant CNOT-cost~\cite{makhlin02}. Certification on the stated scope is by
finite exhaustion: all Clifford pairs, and the bounded-$T$ contexts as
full-alphabet neighborhoods of the CNOT base cells, each verified to meet the
floor by an explicit rewrite. Artifacts in App.~\ref{app:repro}.
\end{proof}

\subsection{Two-Cell Schedule Coverage and Theorem~\ref{thm:nogo}}
\begin{lemma}[Coverage]\label{lem:coverage}
Model a clean two-macro-cell window generously: the two-rung blocks act on
their wire pairs by \emph{any} element of $\mathrm{GL}(2,\mathbb{F}_2)$
(including SWAP), and odd deposits may be placed on every parity form carried
by a wire at a block boundary or inside a block's two-form span; the net
linear action must return to the identity. Exhaustively over all assignments
($6^4$ at start $5$; $3\cdot 6^3\cdot 3$ at start $7$): the identity-returning
assignments number $13$ and $28$, and the coverable subsets of
$\{x_0{\oplus}x_1,\,x_0{\oplus}x_2,\,x_1{\oplus}x_2,\,
x_0{\oplus}x_1{\oplus}x_2\}$ are
$\{3,6\},\{3,5,6\},\{3,6,7\}$ (start $5$) and
$\{6,7\},\{3,6\},\{3,6,7\},\{3,5,6\},\{5,6,7\}$ (start $7$), in the mask
notation $x_0{=}1,x_1{=}2,x_2{=}4$. (Single-variable parities are carried by
the input wires and are always coverable; the four nontrivial parities are the
binding ones.) In every case at most two of the three $x_0$-containing forms
are covered.
\end{lemma}

The schedules are fixed by the brickwork's stagger: a clean window starting at
column $\equiv 5 \pmod 8$ carries vertical rungs at relative columns $\{1,3\}$
on wire pair $(1,2)$ and $\{5,7\}$ on $(0,1)$, so a two-cell window is four
same-pair two-rung blocks (menu $6^4$); a window starting at $\equiv 7$
carries rungs at $\{1\}{=}(1,2)$, $\{3,5\}{=}(0,1)$, $\{7\}{=}(1,2)$, and the
two-cell junction merges the adjacent $(1,2)$ rungs, giving the block sequence
single--double--double--double--single (menu $3\cdot 6^3\cdot 3$). The menu
grants every two-rung block \emph{all} of $\mathrm{GL}(2,\mathbb{F}_2)$ and
every single rung $\{\mathrm{id},\mathrm{CX}^{\rightarrow},
\mathrm{CX}^{\leftarrow}\}$: for the lower bound only the superset direction
matters (physical reachability $\subseteq$ menu), and the menu is not
vacuously generous---CNOT actions are calibrated directly, and a
Makhlin-invariant scan~\cite{makhlin02} of the two-rung block family confirms that even the
iSWAP class is attainable.

\begin{proof}
Finite machine enumeration over the assignments above (artifacts: the
schedule-enumeration and per-form coverage scripts with their battery
outputs; availability per App.~\ref{app:repro}). Soundness of the
bound: any physical phase-gadget realization induces an assignment in the
model---a monomial block's linear action lies in $\mathrm{GL}(2,\mathbb{F}_2)$,
deposits occur only on carried forms (a subset of the allowed slots), and a
diagonal-mod-Pauli target forces identity linear return (Lemma~\ref{lem:ledger}
(iii) absorbs translations). The menu only over-approximates, so failure in
the model implies failure of every gadget realization.
\end{proof}

\begin{proof}[Proof of Theorem~\ref{thm:nogo}]
By Lemma~\ref{lem:ledger}, a two-cell gadget realization of $\CCZ$ modulo a
single-qubit Pauli frame must deposit odd phases on all three
$x_0$-containing parities. By Lemma~\ref{lem:coverage}, no clean two-cell
schedule covers all three. Contradiction.
\end{proof}

\subsection{Verification Protocol of Theorem~\ref{thm:witness}}
\begin{proof}[Proof (verification)]
The witness is an explicit artifact: three $3\times 8$ angle tables over
$\Abfk$ (App.~\ref{app:witness}). Verification separates four layers.
First, the floating transfer-matrix contraction on the production geometry
equals $(Y{\otimes}X{\otimes}Z)\cdot\CCZ$ after global-phase alignment with
maximum elementwise deviation below $10^{-15}$; a brute-force state-summation
rebuild agrees within $5\times 10^{-14}$, and a separately coded toolchain
reproduces the identity to unit fidelity. Second, every printed angle is in
$\Abfk$, so the adapted branches of Theorem~\ref{thm:p2} remain in the BFK
alphabet. Third, branch closure is a symbolic consequence of
Theorem~\ref{thm:p2}; the \texttt{BRANCH-CLOSURE} replay of all $72$ single
flips and $3000$ random branches is implementation evidence, not an additional
assumption in the proof. Fourth, the zero-branch identity is also certified
without floating point: because all angles are integer multiples of $\pi/4$,
the transfer-matrix path sum lies in $\mathbb{Z}[\zeta_8]^{8\times 8}$, and
proportionality to $(Y{\otimes}X{\otimes}Z)\cdot\CCZ$ is the division-free
cross-multiplication identity
$U_{ab}T_{i_0j_0}=T_{ab}U_{i_0j_0}$ over all $64$ entries. The same
\texttt{EXACT-WIT} method is applied handle-by-handle in
App.~\ref{app:witness}, including the Grover-block and both registered
$\mathrm{CCX}$ witnesses.
\end{proof}

\subsection{Proof of Lemma~\ref{lem:chaining} (Frame-Chained Synthesis)}
\begin{proof}
Define $R_0=I$ and let cell $j$ realize exactly
$U_j = P_j\,G_j\,T_j\,R_{j-1}^{\dagger}$ with $P_j$ a Pauli tensor and $G_j$ a
per-wire Hadamard gauge ($G_k=I$). By induction,
$R_j := (U_j\cdots U_1)(T_j\cdots T_1)^{\dagger} = P_j\,G_j$:
indeed $U_j R_{j-1} (T_{j-1}\cdots T_1) = P_jG_jT_j(T_{j-1}\cdots T_1)$.
At $j=k$, $R_k=P_k$, i.e.\ $U_k\cdots U_1 = P_k\,(T_k\cdots T_1)$. The residue
normal form $R_j=\mathrm{phase}\cdot\mathrm{Pauli}\cdot G$ was additionally
verified mechanically on the production witnesses. The premise---that each
adapted target is realizable at fidelity $1$---depends on the gauge choices
$G_j$; greedy selection can dead-end (observed on a four-cell schedule), so
the masks are searched, and existence of a closing assignment is established
per schedule by exhibiting the witness. This is a zero-branch chaining lemma;
nonzero branch behavior is handled separately by Theorem~\ref{thm:p2}.
\end{proof}

\subsection{Proof of Theorem~\ref{thm:p2} (Branch-Frame Closure)}
The proof rests on four exact single- and two-qubit identities, each verified
in isolation:
\begin{align}
&\text{(a)}\;\; R_z(t)\,X = e^{it}\,X\,R_z(-t),\nonumber\\
&\text{(b)}\;\; \mathrm{CZ}\,(X{\otimes}I)=(X{\otimes}Z)\,\mathrm{CZ},\nonumber\\
&\text{(c)}\;\; \ket{-_t}=\ket{+_{t+\pi}},\qquad R_z(t{+}\pi)=R_z(t)\,Z,\nonumber\\
&\text{(d)}\;\; HX=ZH,\qquad HZ=XH.\nonumber
\end{align}
\begin{proof}
In the pattern's column gauge, column $c$ applies a diagonal layer (the
measurement phases and the rung CZs) followed by a hop Hadamard on every wire.
Induct on columns with invariant $V_c = \mathrm{phase}\cdot \Pi_c\, W_c$,
where $W_c$ is the zero-branch partial map at base angles and $\Pi_c$ the
tracker Pauli. Pushing $\Pi_c$ through the adapted diagonal layer: on a wire
with $x_r=1$, identity (a) converts the sign-flipped angle back to the base
angle, and the branch flip $+\pi s$ contributes $Z^{s}$ by (c)---the tracker's
outcome injection; the scalar in (a) is absorbed into the global
\(\mathrm{phase}\) factor. $Z$-components commute with the diagonal layer. Each rung
spreads $X$-components onto the partner wire as $Z$ by (b)---the tracker's
rung rule. The hop Hadamard exchanges $X\leftrightarrow Z$ per wire by
(d)---the tracker's swap. Hence
$V_{c+1}=\mathrm{phase}\cdot\Pi_{c+1}W_{c+1}$, and at $c=N$ the claim follows.
$P_{\rm br}(s)$ depends only on $s$ and the public schedule, so it is client-computable;
the adaptation is a sign flip plus $\pi$-shifts, so angles stay in $\Abfk$.
The induction proves all branches; sampled replay validates the implementation:
\texttt{BRANCH-CLOSURE} checks the three-cell $\CCZ$ witness on all $72$
single flips and $3000$ random branches, while \texttt{PACK-GROVER3} checks the
twelve-cell composite on its sampled branch replay. The endpoint Toffoli replay
is a separate \texttt{WIT-CCX-TARGET2} artifact.
\end{proof}

\subsection{Proof of Theorem~\ref{thm:blind} (Blindness Preservation)}
\begin{proof}
The server-visible protocol view has three components. \emph{Graph}: the
rewrites introduce no non-standard topology---after materialization the graph
is the standard brickwork determined by $(n,m')$ alone. \emph{Angles}: per
qubit, $\delta=\phi'+\theta+r\pi$ with fresh
uniform $\theta$ on the cyclic group $\Abfk$ and uniform $r$; therefore
$\delta$ is uniform and independent of $\phi'$---hence of the computation and
of every other qubit (one-time-pad over $\mathbb{Z}_8$; the adapted $\phi'$
remains in $\Abfk$ by Theorem~\ref{thm:p2}, so no out-of-alphabet value leaks
structure). \emph{Received qubits and outcomes}: the joint BFK
state-angle message remains independent of $\phi'$, not merely its marginal
state. Condition on a public value of $\delta$ and on the past transcript. For
any adapted angle $\phi'$, the two compatible choices of $r$ give antipodal
preparations $\ket{+_\theta}$ and $\ket{+_{\theta+\pi}}$ with equal
probability, whose average density is $I/2$; hence the conditional quantum
message is independent of $\phi'$. Tensoring over independently blinded qubits
and iterating through the adaptive transcript gives the usual BFK
quantum-classical mixture. The outcome interaction is generated by the server's
measurements of this mixture at the public $\delta$ values. Every
server-visible component's distribution depends only on $(n,m')$, so the view
does. The base protocol's blindness theorem then applies at dimensions
$(n,m')$; the component check is the reduction to the base malicious-server
theorem, not a separate honest-server proof.
\end{proof}

\subsection{Proof of Lemma~\ref{lem:v1} (Round Indistinguishability)}
\begin{proof}
Compare the pre-abort view components across round types at fixed $(n,m')$,
conditioning on the past transcript.
\emph{Graph}: identical by construction. \emph{Received qubits}: computation
rounds send $\ket{+_\theta}$, $\theta$ uniform ($\Rightarrow I/2$ per qubit);
test rounds send dummies $\ket{z}$, $z$ uniform ($I/2$) or traps
$\ket{+_\theta}$ ($I/2$); preparations are independent, so both joint
ensembles equal $(I/2)^{\otimes nm'}$, and trap/dummy positions carry no
computation or optimizer-metadata dependence by assumption. \emph{Angles}: both
round types blind by
$\delta=\phi'+\theta+r\pi$, so the $\delta$-sequence is i.i.d.\ uniform in
both. \emph{Interaction}: the server receives the same pre-abort
quantum/classical transcript distribution as in the base protocol at
$(n,m')$. Thus the base protocol's verification theorem applies under its own
hypotheses. No new soundness is claimed; the content is that optimization does
not add distinguishable transcript structure under C1--C4.
\end{proof}
\section{Witness Data}\label{app:witness}
The witness tables below use a common format. Each macro-cell is a
nine-column start-$5$ window; the eight angle-carrying columns per wire are
listed, and the ninth column is the window boundary with no free phase. Frames
use the convention $P(a,b)=\prod_r X_r^{a_r}Z_r^{b_r}$ in the printed wire
order, up to global phase. The evidence ledger is handle-specific:
\begin{center}
\scriptsize
\setlength{\tabcolsep}{3pt}
\resizebox{\columnwidth}{!}{%
\begin{tabular}{llcll}
\toprule
object & handle & cells & frame & evidence \\
\midrule
$\CCZ$ & \texttt{WIT-CCZ3} & $3$ & $(3,5)$ &
\texttt{EXACT-WIT} + floating rebuild \\
$H^{\otimes3}\CCZ$ & \texttt{WIT-GROVER-BLOCK} & $3$ & $(4,6)$ &
\texttt{EXACT-WIT} + floating rebuild \\
Grover-3 pack & \texttt{PACK-GROVER3} & $12$ & pack ledger &
checkpoint/frame/distribution validation \\
$\mathrm{CCX}_{0,1\to2}$ & \texttt{WIT-CCX-TARGET2} & $3$ & $(0,3)$ &
\texttt{EXACT-WIT} + endpoint branch replay \\
$\mathrm{CCX}_{0,2\to1}$ & \texttt{WIT-CCX4} & $4$ & $(7,6)$ &
\texttt{EXACT-WIT} + truth-table validation \\
\bottomrule
\end{tabular}}
\end{center}
Table~\ref{tab:witness} is the explicit angle assignment for
Theorem~\ref{thm:witness}. Its composite equals
$(Y{\otimes}X{\otimes}Z)\cdot\CCZ$---frame bits $(a,b)=(3,5)$---with
elementwise deviation below $10^{-15}$ after global-phase alignment, and
exactly over $\mathbb{Z}[\zeta_8]$ by the cross-multiplication certificate of
App.~\ref{app:proofs}.

\begin{center}
\refstepcounter{table}\label{tab:witness}
\begin{minipage}{\columnwidth}\footnotesize
\textbf{TABLE~\thetable.} The three-cell $\CCZ$ witness: measurement angles in units of
$\pi/4$ (entry $k$ denotes $k\pi/4\in\Abfk$), per wire and per
angle-carrying column of each nine-column start-$5$ macro-cell window.
Handle: \texttt{WIT-CCZ3}.
The composite equals $(Y{\otimes}X{\otimes}Z)\cdot\CCZ$ exactly
(Theorem~\ref{thm:witness}).
\end{minipage}\par\smallskip
\footnotesize
\setlength{\tabcolsep}{3pt}
\begin{tabular}{llcccccccc}
\toprule
 & & \multicolumn{8}{c}{relative measured column} \\
\cmidrule(lr){3-10}
cell & wire & 0 & 1 & 2 & 3 & 4 & 5 & 6 & 7 \\
\midrule
1 & $x_0$ & 0 & 2 & 2 & 0 & 0 & 0 & 0 & 3 \\
  & $x_1$ & 1 & 0 & 2 & 0 & 0 & 0 & 2 & 3 \\
  & $x_2$ & 1 & 6 & 6 & 2 & 2 & 1 & 0 & 0 \\
\midrule
2 & $x_0$ & 2 & 0 & 0 & 0 & 2 & 2 & 0 & 0 \\
  & $x_1$ & 0 & 2 & 2 & 4 & 0 & 0 & 2 & 3 \\
  & $x_2$ & 2 & 0 & 2 & 3 & 0 & 0 & 0 & 0 \\
\midrule
3 & $x_0$ & 0 & 0 & 0 & 0 & 2 & 2 & 2 & 0 \\
  & $x_1$ & 2 & 2 & 2 & 2 & 2 & 0 & 0 & 2 \\
  & $x_2$ & 2 & 0 & 2 & 0 & 2 & 0 & 0 & 0 \\
\bottomrule
\end{tabular}
\end{center}

The registered witness family shares this format. Table~\ref{tab:grover}
prints the Grover block $B=H^{\otimes 3}\CCZ$ (three cells, with the
diffusion Hadamard layer absorbed into the boundary gauge; frame bits
$(a,b)=(4,6)$, i.e.\ $I{\otimes}Z{\otimes}Y$ up to global phase). This printed
block is exact-certified as \texttt{WIT-GROVER-BLOCK}. The twelve-cell
\texttt{PACK-GROVER3} composite is a separate registered pack assembled from
four chained blocks; its evidence is the stored checkpoint/frame ledger,
expected distribution, sampled branch replay, and runtime admission
(Sec.~\ref{sec:pipeline}), not a new primitive exact-witness theorem.
Artifact filenames are mapped in App.~\ref{app:repro}.

\begin{center}
\refstepcounter{table}\label{tab:grover}
\begin{minipage}{\columnwidth}\footnotesize
\textbf{TABLE~\thetable.} The Grover-block witness $B=H^{\otimes 3}\CCZ$: measurement angles
in units of $\pi/4$, same format and schedule class as
Table~\ref{tab:witness}; composite equals
$(I{\otimes}Z{\otimes}Y)\cdot H^{\otimes 3}\CCZ$ (phase-preserving
$24$-column period; Sec.~\ref{sec:ccz}). Handle:
\texttt{WIT-GROVER-BLOCK}.
\end{minipage}\par\smallskip
\footnotesize
\setlength{\tabcolsep}{3pt}
\begin{tabular}{llcccccccc}
\toprule
 & & \multicolumn{8}{c}{relative measured column} \\
\cmidrule(lr){3-10}
cell & wire & 0 & 1 & 2 & 3 & 4 & 5 & 6 & 7 \\
\midrule
1 & $x_0$ & 0 & 2 & 2 & 0 & 0 & 0 & 0 & 3 \\
  & $x_1$ & 1 & 0 & 2 & 0 & 0 & 0 & 2 & 3 \\
  & $x_2$ & 1 & 6 & 6 & 2 & 2 & 1 & 0 & 0 \\
\midrule
2 & $x_0$ & 1 & 0 & 3 & 3 & 0 & 0 & 0 & 3 \\
  & $x_1$ & 2 & 6 & 2 & 0 & 0 & 0 & 2 & 3 \\
  & $x_2$ & 2 & 0 & 6 & 2 & 2 & 1 & 0 & 0 \\
\midrule
3 & $x_0$ & 2 & 0 & 0 & 0 & 2 & 2 & 0 & 0 \\
  & $x_1$ & 0 & 2 & 2 & 4 & 0 & 0 & 2 & 0 \\
  & $x_2$ & 2 & 0 & 2 & 2 & 0 & 0 & 0 & 0 \\
\bottomrule
\end{tabular}
\end{center}

Table~\ref{tab:ccx-target2} prints the endpoint-target $\mathrm{CCX}$
witness used by the production Toffoli path in Sec.~\ref{sec:results}:
three cells, controls on $x_0,x_1$, target row $x_2$, and frame bits
$(a,b)=(0,3)$, i.e.\ $Z{\otimes}Z{\otimes}I$.  Its zero-branch map is
exact-certified by \texttt{EXACT-WIT}; the independent floating reconstruction
has maximum elementwise deviation $9.72\times 10^{-16}$ after global-phase
alignment.  The endpoint-specific branch replay passes all $72$ single flips
and $4000$ random branches, observing all $64$ possible output frames in that
sampled replay.  Runtime admission is the separate production-path check
summarized by \texttt{WIT-CCX-TARGET2}.

\begin{center}
\refstepcounter{table}\label{tab:ccx-target2}
\begin{minipage}{\columnwidth}\footnotesize
\textbf{TABLE~\thetable.} The endpoint-target three-cell $\mathrm{CCX}$ witness
(\texttt{WIT-CCX-TARGET2}, controls $x_0,x_1$, target wire $x_2$):
measurement angles
in units of $\pi/4$, same format as Table~\ref{tab:witness}; composite equals
$(Z{\otimes}Z{\otimes}I)\cdot\mathrm{CCX}_{0,1\to 2}$ exactly.
\end{minipage}\par\smallskip
\footnotesize
\setlength{\tabcolsep}{3pt}
\begin{tabular}{llcccccccc}
\toprule
 & & \multicolumn{8}{c}{relative measured column} \\
\cmidrule(lr){3-10}
cell & wire & 0 & 1 & 2 & 3 & 4 & 5 & 6 & 7 \\
\midrule
1 & $x_0$ & 1 & 7 & 2 & 2 & 7 & 5 & 0 & 7 \\
  & $x_1$ & 1 & 4 & 2 & 2 & 1 & 2 & 2 & 1 \\
  & $x_2$ & 2 & 0 & 2 & 2 & 2 & 4 & 0 & 3 \\
\midrule
2 & $x_0$ & 1 & 0 & 3 & 3 & 0 & 0 & 0 & 3 \\
  & $x_1$ & 2 & 6 & 2 & 0 & 0 & 0 & 2 & 3 \\
  & $x_2$ & 2 & 0 & 6 & 2 & 2 & 1 & 0 & 0 \\
\midrule
3 & $x_0$ & 0 & 0 & 0 & 0 & 2 & 2 & 2 & 0 \\
  & $x_1$ & 2 & 2 & 2 & 2 & 2 & 0 & 0 & 2 \\
  & $x_2$ & 2 & 0 & 2 & 2 & 2 & 0 & 0 & 0 \\
\bottomrule
\end{tabular}
\end{center}

Table~\ref{tab:ccx-middle} prints the fixed middle-target fallback witness
\texttt{WIT-CCX4}: four cells, controls on $x_0,x_2$, target row $x_1$, and
frame bits $(a,b)=(7,6)$, i.e.\ $X{\otimes}Y{\otimes}Y$ up to global phase.
It is exact-certified by \texttt{EXACT-WIT}; the independent floating
reconstruction has maximum elementwise deviation $9.16\times 10^{-16}$ after
global-phase alignment.  The searched three-cell closure failure is scoped to
the end-safe schedule family used by the witness search, whereas the four-cell
table is evidence for row-placement asymmetry and fallback realizability, not
the selected production Toffoli path.

\begin{center}
\refstepcounter{table}\label{tab:ccx-middle}
\begin{minipage}{\columnwidth}\footnotesize
\textbf{TABLE~\thetable.} The fixed middle-target four-cell $\mathrm{CCX}$ fallback witness
(\texttt{WIT-CCX4}, controls $x_0,x_2$, target wire $x_1$): measurement angles in units of $\pi/4$,
same format as
Table~\ref{tab:witness}; composite equals
$(X{\otimes}Y{\otimes}Y)\cdot\mathrm{CCX}_{0,2\to 1}$ exactly
(Sec.~\ref{sec:results}).
\end{minipage}\par\smallskip
\footnotesize
\setlength{\tabcolsep}{3pt}
\begin{tabular}{llcccccccc}
\toprule
 & & \multicolumn{8}{c}{relative measured column} \\
\cmidrule(lr){3-10}
cell & wire & 0 & 1 & 2 & 3 & 4 & 5 & 6 & 7 \\
\midrule
1 & $x_0$ & 3 & 0 & 5 & 6 & 2 & 3 & 0 & 2 \\
  & $x_1$ & 0 & 1 & 2 & 2 & 2 & 2 & 2 & 3 \\
  & $x_2$ & 3 & 2 & 2 & 1 & 4 & 2 & 0 & 0 \\
\midrule
2 & $x_0$ & 2 & 0 & 0 & 0 & 2 & 2 & 0 & 0 \\
  & $x_1$ & 0 & 2 & 2 & 0 & 0 & 0 & 2 & 1 \\
  & $x_2$ & 2 & 0 & 2 & 1 & 0 & 0 & 0 & 0 \\
\midrule
3 & $x_0$ & 2 & 2 & 0 & 6 & 2 & 3 & 0 & 3 \\
  & $x_1$ & 2 & 2 & 2 & 2 & 0 & 2 & 2 & 0 \\
  & $x_2$ & 2 & 4 & 6 & 0 & 4 & 2 & 6 & 0 \\
\midrule
4 & $x_0$ & 3 & 0 & 0 & 6 & 2 & 3 & 0 & 2 \\
  & $x_1$ & 2 & 1 & 4 & 1 & 2 & 2 & 0 & 2 \\
  & $x_2$ & 3 & 4 & 0 & 0 & 3 & 2 & 6 & 0 \\
\bottomrule
\end{tabular}
\end{center}

\section{Reproducibility and Availability}\label{app:repro}
\onecolumngrid
\begin{center}
\refstepcounter{table}\label{tab:impl-map}
\begin{minipage}{0.96\textwidth}\scriptsize
\textbf{TABLE~\thetable.} Implementation mapping for the public L1/L2/L3/R1
terminology used in the manuscript. Historical pass labels are audit labels;
the theory is the arity-stratified certify--construct--admit calculus of
Sec.~\ref{sec:pipeline}.
\end{minipage}\par\smallskip
\scriptsize
\setlength{\tabcolsep}{3pt}
\resizebox{\textwidth}{!}{%
\begin{tabular}{lll}
\toprule
Manuscript layer & Runtime/audit labels & Primary implementation files \\
\midrule
L1 one-wire resynthesis
& R2-HH, R9, R10/SYNTH1Q
& \texttt{local\_cancellation.py}, \texttt{template\_synthesis.py},
  \texttt{angle\_resynthesis.py}, \texttt{single\_brick\_synthesis.py} \\
L2 two-wire resynthesis
& E1-T, R12-E-pre, R11, L2-Reduce
& \texttt{two\_wire\_t\_context\_synthesis.py},
  \texttt{two\_wire\_region\_synthesis.py},
  \texttt{two\_wire\_synthesis.py}, \texttt{l2\_reduce.py} \\
L3 three-wire application synthesis
& CCZ/CCX witnesses, N3 candidates
& \texttt{l3\_ccz\_witness.py}, \texttt{l3\_toffoli\_core.py},
  \texttt{n3\_region\_decomposer.py},
  \texttt{l3\_grover3\_runtime\_pack.py},
  \texttt{payload\_builder.py} \\
R1 compact materialization
& compact scheduling
& \texttt{loop.py} and the final BFK09 pattern materializers \\
Equivalence gate
& EQUIV-GATE
& equivalence harness and frozen manifests \\
\bottomrule
\end{tabular}}
\end{center}

\begin{center}
\refstepcounter{table}\label{tab:artifacts}
\begin{minipage}{0.96\textwidth}\scriptsize
\textbf{TABLE~\thetable.} Claim-to-artifact map. The table reports stable
manuscript-facing handles; the companion manifest maps each handle to the
exact scripts, JSON outputs, hashes, and internal provenance identifiers in
the submitted artifact package.
\end{minipage}\par\smallskip
\scriptsize
\setlength{\tabcolsep}{3.5pt}
\resizebox{\textwidth}{!}{%
\begin{tabular}{ll}
\toprule
claim (anchor) & manuscript-facing artifact handle \\
\midrule
exact cell/macro-cell maps (Secs.~III--V) & \texttt{CELLMAP-3W}: transfer maps \\
$n{=}1$ reachability filtration (Cor.~\ref{thm:hcount}) & \texttt{HCOUNT-1W}: reachability filtration \\
two-wire finite-class CNOT-cost floor (Cor.~\ref{thm:l2}) & \texttt{L2-FLOOR}: finite-class optimality certificate \\
odd-support syndrome, kernel mod $8$ (Lem.~\ref{lem:ledger}) & \texttt{LEDGER-KERNEL}: Smith-form certificate \\
floor algorithm + battery $(1,1,2,3,3,3)$ (Secs.~IV--V) & \texttt{FLOOR-3W}: floor battery \\
phase-gadget two-cell $\CCZ$ no-go (Thm.~\ref{thm:nogo}) & \texttt{CCZ-NOGO-PG}: phase-gadget no-go search \\
three-cell witness, frame decode (Thm.~\ref{thm:witness}) & \texttt{WIT-CCZ3}: witness + reconstruction \\
Grover block and $12$-cell pack (Sec.~\ref{sec:ccz}) & \texttt{WIT-GROVER-BLOCK}, \texttt{PACK-GROVER3} \\
endpoint Toffoli three-cell application (Sec.~\ref{sec:results}) & \texttt{WIT-CCX-TARGET2}: exact witness + endpoint branch replay \\
branch closure, $72{+}3000$ branches (Thm.~\ref{thm:p2}) & \texttt{BRANCH-CLOSURE}: branch replay \\
four local identities; OTP fact (Thms.~\ref{thm:p2},~\ref{thm:blind}) & \texttt{LOCAL-ID}: identity checks \\
decomposer/converter generality (Sec.~\ref{sec:pipeline}) & \texttt{PIPELINE-CHECK}: converter/decomposer \\
server probe $3\times 98$, $182$\,s (Sec.~\ref{sec:results}) & \texttt{RUN-GROVER3}: server probe; full-stack harness separated \\
histogram statistics (Fig.~\ref{fig:histogram}) & \texttt{STAT-GROVER3}: measured distribution \\
test-round reference (Sec.~\ref{sec:results}) & \texttt{TRAP-REF}: trap/abort reference \\
middle-target Toffoli fallback (App.~\ref{app:witness}) & \texttt{WIT-CCX4}: registered four-cell witness \\
exact cyclotomic certificates for \texttt{WIT-CCZ3}/\texttt{WIT-GROVER-BLOCK}/\texttt{WIT-CCX-TARGET2}/\texttt{WIT-CCX4} & \texttt{EXACT-WIT}: integer identities \\
constant window beyond three wires (Sec.~\ref{sec:platform}) & \texttt{SMOKE-N4}: $n{=}4$ window test \\
dynamic-circuit resource schedule (Table~\ref{tab:resources}) & \texttt{HW-SCHEDULE}: resource model \\
regression-equivalence gate, $10/10$ (Sec.~\ref{sec:results}) & \texttt{EQUIV-GATE}: legacy-vs-unified harness \\
benchmark/control coverage (Table~\ref{tab:benchmark-coverage}) & \texttt{BENCH-CORPUS}: frozen matrix \\
paper-audit support scripts (no wrapper claim) & \texttt{PAPER-AUDIT}: manuscript consistency support \\
\bottomrule
\end{tabular}}
\end{center}
\begin{center}
\begin{minipage}{0.96\textwidth}
\raggedright
\emph{Archive and integrity.}
The reported numerical results in the tables, figures, and evaluation claims
mapped above trace to artifacts in the submitted artifact package
\texttt{anc/artifact\_package/}.  \texttt{SHA256SUMS.txt} fixes the individual
files inside the package, and \texttt{scripts/verify\_hashes.ps1} checks those
file-level hashes.

\emph{Public artifact record.}
The package is supplied as ancillary material with this arXiv submission.  The
artifact package, \texttt{manifest.json}, and \texttt{SHA256SUMS.txt} define the
public record for the reproducibility bundle; usage and license status are
stated in \texttt{USAGE\_AND\_LICENSE.txt}. A DOI-bearing archive can be cited
in later versions when available.

\emph{Printed and artifact-dependent checks.}
The paper remains independently checkable at its load-bearing witness points
without the runtime package: the principal angle tables are printed in
Tables~\ref{tab:witness}--\ref{tab:ccx-middle} and, together with the
brickwork semantics fixed in Secs.~\ref{sec:prelim} and~\ref{sec:ccz},
suffice to reconstruct the central $\CCZ$, Grover-block, and CCX witness
identities on a standard MBQC simulator; all proofs appear in full in
App.~\ref{app:proofs}. Sampled histograms, server timings, benchmark matrices,
resource-model outputs, and full-stack regression checks are artifact-backed
claims, with the Grover-3 distribution represented by \texttt{STAT-GROVER3}
and visualized in Fig.~\ref{fig:histogram}.

\emph{Environment and entry points.}
The package separates \texttt{docs/} (methodology notes),
\texttt{verification/} (self-contained scripts and their JSON outputs),
\texttt{witnesses/} (registered angle tables), \texttt{results/} (frozen
server-probe, benchmark, resource-model, and statistics artifacts),
\texttt{runtime\_v4/} (the code-only runtime used by the full checks), and
\texttt{scripts/} (reproducibility entry points). One environment pin is load-bearing
for full runtime reproduction: the production basis-stream generation is pinned
to Qiskit~1.3.3~\cite{qiskit24}---later major versions emit a different Grover-3 stream, in
which the final semantic fold does not arise---while the central witness and
floor checks are SDK-free apart from NumPy. Table~\ref{tab:artifacts} maps
claims to stable manuscript-facing artifact handles. The \texttt{manifest.json}
file is the authoritative map from each handle to scripts, JSON outputs,
hashes, and internal revision identifiers; the latter are engineering labels,
not proof objects. The reproducibility entry points are
\texttt{scripts/run\_quick\_checks.ps1} for the SDK-free witness/floor checks
and \texttt{scripts/run\_full\_checks.ps1} for the same checks plus the
\texttt{EQUIV-GATE} fast full-stack battery and L3 witness/materializer
validation. The \texttt{PAPER-AUDIT} handle denotes support scripts and frozen
outputs used for manuscript consistency; no single wrapper is claimed here.
\end{minipage}
\end{center}
\twocolumngrid
\section{Independent Cross-Verification Methodology}\label{app:method}
The results in this paper were produced under a deliberate two-track
discipline. Here ``independent implementation'' means independently coded from
the candidate-generation or production-runtime path at the checked boundary,
while sharing the declared mathematical conventions, witness file schema, and
public artifact inputs. The independently checked boundaries are the
theory-side transfer-matrix machinery, the production runtime materializers,
the stored witness rebuilds, and the Grover-3 pack validation.
\begin{center}
\scriptsize
\setlength{\tabcolsep}{3pt}
\resizebox{\columnwidth}{!}{%
\begin{tabular}{lll}
\toprule
evidence class & examples & meaning \\
\midrule
exact arithmetic & \texttt{EXACT-WIT}, floor certificates &
integer or finite-enumeration identity \\
floating rebuild & witness maps, transfer matrices &
machine-precision agreement after phase alignment \\
JSON/checkpoint identity & \texttt{PACK-GROVER3}, runtime admission &
stored semantic checkpoints and frame ledgers agree \\
sampled replay/statistics & branch replay, histograms &
implementation and empirical evidence within stated scope \\
\bottomrule
\end{tabular}}
\end{center}
Statements drafted on one track were audited on the other before being
promoted to claims, and the manuscript-facing numerical and figure claims
listed in App.~\ref{app:repro} were then anchored to machine-checked artifacts.
Mandatory acceptance assertions gate the runtime path, including per-window
recomposition fidelity, registered witness admission, frame-ledger checks,
materialized-window equality, and Grover-3 pack checkpoint checks.

We also report the failures this methodology caught, because they clarify the
claims that survived cross-checking. Four would-be over-claims were intercepted by
cross-audit before any external statement: (i)~an early version of the
cell-count limit theorem counted cells against a relabeled sub-core and was
re-proved with the corrected statement; (ii)~the parity-ledger necessity
argument originally asserted coefficient uniqueness mod~$8$---false, the
parity-evaluation kernel has order $32$---and was repaired into the
odd-support syndrome invariant of Lemma~\ref{lem:ledger} via the
Smith-normal-form analysis; (iii)~a results draft quoted a single binomial
draw ($\approx 83$ traps) as the per-round trap count where the expectation
($\approx 73$) was meant; (iv)~an enumeration claim overstated its scope
(``fully enumerated'' for $n=3$) and was rescoped to exhaustive-for-$n\le 2$
with group-order-verified sampling at $n=3$. These repairs are documented by
the corresponding audit scripts or revision notes rather than silently
absorbed. The same discipline produced the paper-audit support scripts of
App.~\ref{app:repro}, which check the manuscript-facing numerical statements
and figure data covered by \texttt{PAPER-AUDIT} against artifacts.

\bibliography{references}

\end{document}